\documentclass[journal]{IEEEtran}

\usepackage{
  amsmath,
  amssymb,
  amsthm,
  mathtools,
  enumitem,
  xcolor,
  graphicx,
}
\usepackage[hidelinks]{hyperref}
\usepackage[nameinlink,capitalise,noabbrev]{cleveref}
\usepackage[T1]{fontenc}

\theoremstyle{plain} 
\newtheorem{theorem}{Theorem}[section] 
\newtheorem*{theorem*}{Theorem} 
\newtheorem{lemma}[theorem]{Lemma}

\newtheorem{corollary}[theorem]{Corollary}

\theoremstyle{definition} 
\newtheorem{definition}[theorem]{Definition}

\usepackage{cite}

\ifCLASSINFOpdf
\else
\fi
\begin{document}
%
\title
{
  Underwater Color Restoration \\ with Vanishing Uncertainty
}
%
%
%

\author{
  Grigory~Solomatov,
  Derya~Akkaynak$^{\ast}$ \\
  \small Hatter Department of Marine Technologies, University of Haifa, Haifa, Israel \\
  \small Interuniversity Institute for Marine Sciences, Eilat, Israel \\
  \small$^\ast$Corresponding author. Email: dakkaynak@univ.haifa.ac.il\and
}

\maketitle


\begin{abstract}
  Underwater color restoration promises to unlock color as a reliable signal for aquatic sciences,
  but achieving this with scientific confidence remains out of reach.
  Current methods are validated almost exclusively on an empirical basis,
  which provides confidence only to the extent that the vast diversity
  of possible visibility conditions is covered with
  end-to-end testing using a known ground truth.
  This is exacerbated by color restoration being a fatally ill-posed problem
  when considered in full mathematical generality,
  requiring additional constraints to narrow the solution to a finite uncertainty interval.
  The gap between which constraints suffice in theory
  and which constraints are satisfied by real-world data
  is poorly understood,
  making it unclear whether existing methods are solving a problem that is actually solvable.
  In this article, we investigate the theoretical side of this gap,
  identifying idealized conditions which guarantee
  bounded uncertainty that converges to zero
  as the spatial resolution of the camera increases.  
\end{abstract}


%
\IEEEpeerreviewmaketitle

\section{Introduction}

Underwater color restoration
is the problem of ``removing water'' from underwater images,
reversing color distortions caused by light traveling through a scattering medium.
The first relevant techniques were conceived for image dehazing in air
\cite{nayar1999vision,schechner2001instant},
and were subsequently adapted to underwater \cite{schechner2004clear,schechner2005recovery}.
As this technology developed,
the original motivation of enhancing visibility
evolved into a promise of unlocking a fundamentally new way of collecting data about our environment \cite{akkaynak_revised_2018,akkaynak2019sea}.
Scientific use of color was already prevalent in air \cite{rossel2006colour,richardson2007use,ahrends2008quantitative,sonnentag2012digital}
and had seen some success underwater \cite{winters2009photographic,beijbom2012automated},
though not without color distortions as a central challenge.
The ability to digitally remove water from images promised to overcome this,
enabling use of color as a signal in aquatic sciences.
Detection of bleaching in corals,
estimation of chlorophyll concentration in seagrass meadows,
as well as species classification of fish and other organisms
are all examples of tasks that could benefit from color restoration.
In this regard, a prudent question is whether
all of this can be achieved with scientific accuracy and precision.

\begin{figure}[ht]
\centering
\includegraphics[width=0.9\linewidth]{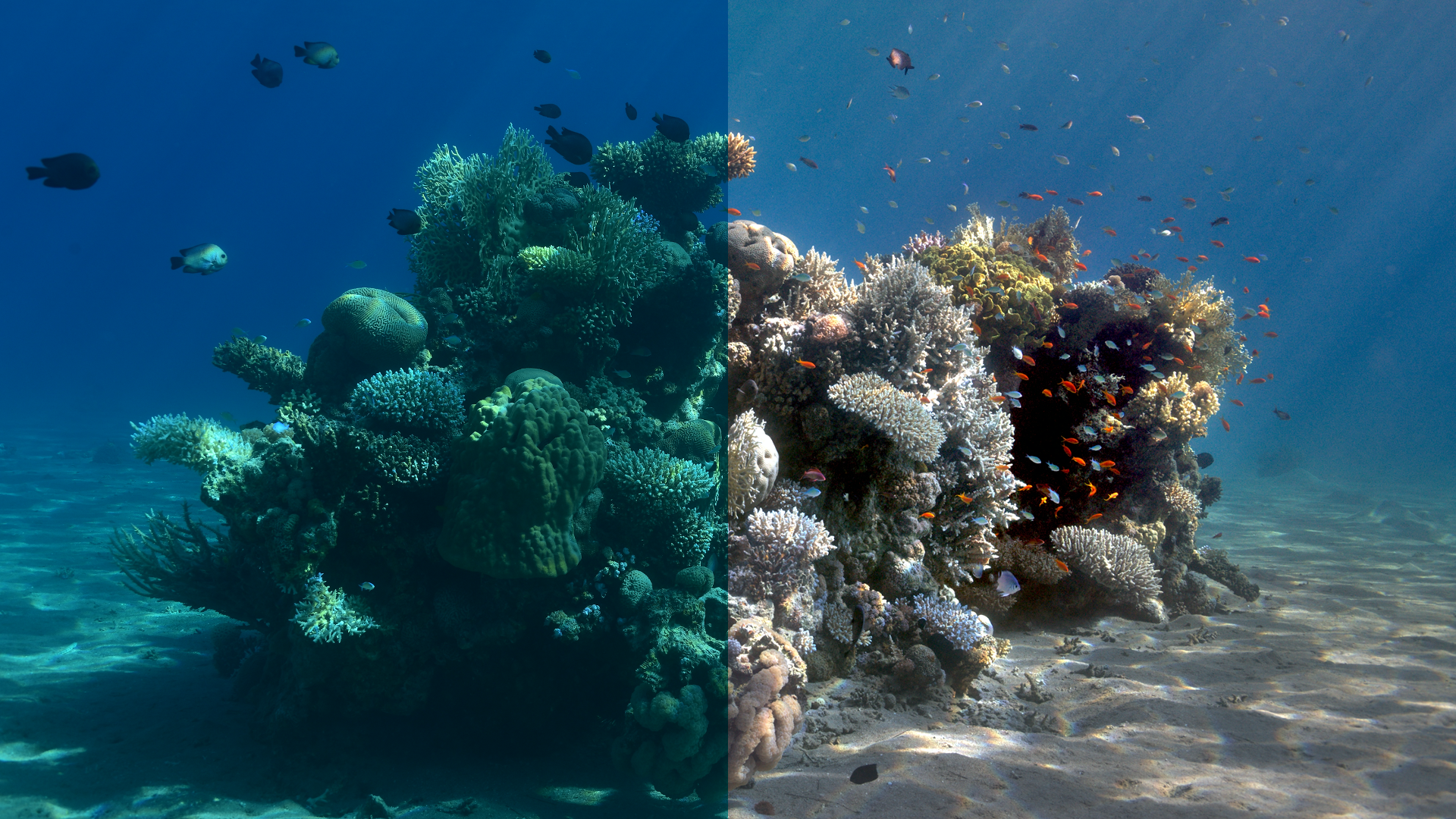}
\caption{Underwater image of a coral reef. Left: original. Right: SeaThru \cite{akkaynak2019sea}.}
\label{fig:seathru}
\end{figure}

\subsection{Scientific confidence}

The principled approach to underwater color restoration
is via the radiative transfer equation (RTE),
which models light propagation in scattering media.
Its general form is far too complicated and impractical for the task,
but even crude simplifications can produce impressive results,
as seen in \cref{fig:seathru}.
The main challenge for physics-based methods
is the underconstrained nature of the resulting inverse problems,
with infinitely many substantially different solutions
being equally consistent with any input data,
at least in the case of full mathematical generality.
Unconstrained color recovery is therefore
a truly impossible problem in the most fundamental sense,
far beyond that of engineering practicalities.

Fortunately, real-world data is not unconstrained,
as evidenced by, e.g., the well-known fact that
neighboring pixels are always correlated in ``natural'' images.
However, whether and when real-world images are sufficiently constrained
for meaningful color restoration remains poorly understood,
limiting scientific confidence in the output of \emph{all} existing methods.
This is not to say that all current methods are inaccurate or imprecise,
but that their accuracy and precision can only be assessed
empirically on a case-by-case basis, as seen in e.g. \cite{alsakar2025underwater}.
The issue here goes beyond
deployment in conditions distinct from those of validation\textemdash
which is itself a real concern
due to the strong spatio-temporal variation of water's optical properties \cite{jerlov1976marine}.
Instead, if the experimental setup happens to instantiate an ill-posed problem,
then the prediction error can only be estimated
to the extent that the solution space clusters around the ground truth.
In the worst case,
this space can be unbounded (\cref{thm:ill-posedness}),
making error estimation fundamentally impossible
without knowing said ground truth.
Statistical error estimation is of course still possible,
but only if the deployment conditions match those of validation.

\subsection{Feasibility gap}

Regardless of method design,
our scientific confidence in underwater color restoration
is bottlenecked by our ability to answer
the following key questions about when the problem is meaningfully solvable:
\begin{enumerate}
  \item Which conditions are sufficient and which are necessary?
  \item Which conditions are satisfied by real-world data?
\end{enumerate}
The gap between the answers to these
determines the extent to which color restoration is ultimately feasible;
however, both remain largely unanswered to date.
The former calls for a theoretical investigation,
for which this article is an early contribution,
while the latter requires broad-scale empirical data acquisition and analysis,
the basis of which is formed by existing empirical work on color restoration.
For efficient progress,
these two pursuits would benefit from being mutually informed by one another,
guiding future research questions with the aim of narrowing the known feasibility gap.

\subsection{Contributions}

In this article, we identify a set of sufficient conditions
for color restoration to be possible.
More concretely, we show that
when these conditions are satisfied,
the correctly restored pixel values
can be bounded to finite uncertainty intervals
whose lengths vanish as the spatial resolution of the camera increases.
As this is early work on the theoretical requirements for color restoration,
our conditions are idealized beyond immediate applicability to real-world images,
though the obtained results help narrow the current gap in understanding
between what is sufficient in theory and what is granted in reality,
i.e., the aforementioned feasibility gap.
Moreover, these results improve on our previous work in \cite{solomatov2026conditions},
where the assumed setting was idealized to a strictly greater degree.
The most important changes are summarized below.

Firstly, instead of assuming a perfect hyperspectral camera
capable of flawlessly capturing apparent radiance,
our current work makes no assumptions on the camera whatsoever.
Any camera is supported in principle, even monochromatic;
and the spectral sensitivity functions need not be known.
Secondly, our previous work relied on oracle access to
inherent radiance segmentation\textemdash
a partitioning of image pixels into subsets sharing the same inherent radiance.
Although we rely on it here as well,
we also identify sufficient conditions for such segmentation
to be uniquely determinable by an algorithm.
It is with these conditions that we depart from a realistic setting the most
as we assume that the inherent radiance is piecewise constant in image coordinates.
Understanding this idealized setting, however,
suggests a clear direction for relaxing these assumptions in future research,
advancing the pursuit towards underwater color restoration with scientific confidence.

\subsection{Problem Statement}

The mathematical question we consider can be informally stated as follows:
Given an underwater photograph of e.g. a fish at a known distance \(z_1\) from the camera,
what would its color appear to be at some other distance \(z\)?
The interesting case in practice is \(z=0\),
where the water no longer impacts the color except via the ambient light,
but mathematically it hardly makes a difference which target distance \(z\) we consider.
To state this question more formally
requires an image formation model, which we take to be:
\begin{align}
  \label{eq:imf}
  P_{i,j}(z_i) = \textstyle\int_\Lambda(e^{-cz_i}L_i + (1 - e^{-cz_i})B)S_j\mathrm{d}\lambda
  \ .
\end{align}
Here,
\(P_{i,j}(z_i)\) denotes the pixel intensity at pixel index \(i\)
corresponding to distance \(z_i\) from the camera,
with \(j\) denoting the channel index, e.g., red, green or blue.
Furthermore,
\(\Lambda\) is the interval on the real number-line \(\mathbb{R}\)
corresponding to wavelengths of light in the visible spectrum.
In the integrand,
the quantities \(c,L_i,B, S_j\) are all functions of wavelength \(\lambda \in \Lambda\),
and they have the following physical interpretations:
The spectral sensitivity functions \(S_j(\lambda)\)
describe how sensitive the camera is to incoming light.
The inherent radiance \(L_i(\lambda)\) describes the spectral radiance towards the camera at the fish's surface.
The beam attenuation coefficient \(c(\lambda)\) describes how quickly light attenuates with distance in the water.
Finally, the backscatter at infinity \(B(\lambda)\)
is incoming radiance in the direction of the camera
when looking at ``infinity''\textemdash
a line of sight containing only water as far as the eye (or camera) can see.
Finally, \eqref{eq:imf} is only valid when looking horizontally,
and even then is derived from the RTE under simplifying assumptions
(see \cref{sec:imf}).
What we mean by color restoration in this setting
is the problem of estimating \(P_{i,j}(z)\) from \(P_{i,j}(z_i)\) and \(z_i\).
In general, this is ill-posed (see \cref{sec:ill-posedness}),
so additional assumptions are necessary.

\subsection{Assumptions and results}

Our results consist of two main parts, both centered around inherent radiance segmentation\textemdash
a partition of all pixel indices into sets \(\mathcal{I}\)
satisfying \(L_i = L_\mathcal{I}\) for all \(i \in \mathcal{I}\),
where \(L_\mathcal{I}\) is the inherent radiance shared by all pixels in \(\mathcal{I}\).
In \cref{sec:using-irs}, we rely on this segmentation to estimate \(P_{i, j}(z)\),
while in \cref{sec:obtaining-irs}, we show how this segmentation may be obtained in an idealized setting.
The highlights are presented below.

For the first main result, in addition to having
the pixel intensities \(P_{i,j}(z_i)\)
and the associated distances \(z_i\),
we require \emph{a priori} upper and lower bounds on \(c\) across \(\Lambda\),
i.e., \(u_1, u_2 \in \mathbb{R}\) satisfying \(u_1 \le e^{-c(\lambda)} \le u_2\) for all \(\lambda \in \Lambda\).
Furthermore, we require at least one pixel in the image
to be associated with a very large (infinite) distance,
and we denote the corresponding pixel intensity by \(P_j(\infty)\) for every channel \(j\).
Finally, we assume that some inherent radiance segmentation is known for pixels with finite distance.

\begin{theorem*}[\ref{thm:main1-eval}]
  Let  \(u_1,u_2 \in \mathbb{R}\) such that
  \(u_1 \le e^{-c(\lambda)} \le u_2\) for all \(\lambda \in \Lambda\),
  and let \(\mathcal{I}\) be an inherent radiance segment.  
  For every \(k \in \mathcal{I}\), let \(\alpha_k:\mathbb{R} \to \mathbb{R}\) be any function, and let
  \begin{align*}
    \Psi_\mathcal{I}(z)
    &= \max_{u_1 \le u \le u_2}|\textstyle\sum_{k \in \mathcal{I}} \alpha_k(z)u^{z_k} - u^z| \ .
  \end{align*}
  Writing \(\bar{P}_{i,j}(z) = P_{i,j}(z) - P_j(\infty)\) for all \(i \in \mathcal{I}\),
  we then have
  \begin{align*}
    \bar{P}_{i,j}(z)
    &= \textstyle\sum_{k \in \mathcal{I}}\alpha_k(z)\bar{P}_{k,j}(z_k)
      \pm \Psi_\mathcal{I}(z)\textstyle\int_\Lambda |L_\mathcal{I} - B| S_j\mathrm{d}\lambda \ ,
  \end{align*}
  where the notation \(x = y \pm \delta\) means \(x \in \{y + w : |w| \le |\delta|\}\).
\end{theorem*}

The intuitive meaning of \cref{thm:main1-eval} can be understood as follows:
Within each segment \(\mathcal{I}\),
the user may compute the values \(\alpha_k(z)\) that minimize \(\Psi_\mathcal{I}(z)\).
For any \(i \in \mathcal{I}\), this immediately yields the approximation
\begin{align*}
  P_{i,j}(z) \approx \textstyle\sum_{k \in \mathcal{I}}\alpha_k(z)\bar{P}_{k,j}(z_k) + P_j(\infty)
\end{align*}
with the error bounded by \(\Psi_\mathcal{I}(z)\textstyle\int_\Lambda |L_\mathcal{I} - B| S_j\mathrm{d}\lambda\).
Furthermore, if we increase the resolution of our camera,
thereby increasing the cardinality of each segment \(\mathcal{I}\),
then in the limit \(|\{z_i\}_{i \in \mathcal{I}}| \to \infty\) we also have \(\Psi_\mathcal{I}(z) \to 0\),
provided \(\sum_{i \in \mathcal{I}}z_i^{-1} \to \infty\); this follows directly from
the Müntz–Szász theorem in approximation theory.
In other words, \(P_{i, j}(z)\) can be recovered to arbitrary precision
as long as the segment \(\mathcal{I}\) containing \(i\)
has enough distinct distances.

The second main result is concerned with
obtaining inherent radiance segmentation,
and it requires making assumptions about
the continuous version of the image,
which we call the profile.
Since, for any segment \(\mathcal{I}\) and any \(i \in \mathcal{I}\),
\begin{align*}
  \bar{P}_{i,j}(z)
  = P_{i,j}(z) - P_j(\infty) = \textstyle\int_\Lambda e^{-cz}(L_\mathcal{I} - B)S_j \mathrm{d}\lambda \ ,
\end{align*}
it is convenient to model the profile as a function
\begin{align*}
  p: (x, y) \mapsto \textstyle\int_\Lambda e^{-c(\lambda)z(x, y)}Q[x, y](\lambda)\mathrm{d}\lambda
  \ ,
\end{align*}
where \((x, y) \in [0, 1]^2\) is any point in image coordinates,
while \(z: [0, 1]^2 \to [0, \infty)\) is the scene distance map,
and \(Q[x, y]: \Lambda \to \mathbb{R}\) is some function of wavelength.
For any point \((x, y)\), consider the region
\begin{align*}
  \mathcal{H}(x, y) = \{(x', y') \mid Q[x', y'] = Q[x, y]\}
\end{align*}
under the subspace topology from \(\mathbb{R}^2\).
We say that \(p\) is \emph{connected} if
\(\mathcal{H}(x, y)\) is always connected,
i.e., it is not a union of two disjoint nonempty open subsets.
We say that \(p\) is \emph{discrete} if \(\mathcal{H}(x, y)\) is always semi-open,
i.e., its interior is dense in its closure.
Finally, for any \(\varepsilon > 0\),
we say that \(p\) is \(\varepsilon\)-\emph{separable} if
\begin{align*}
  Q[x_1, y_1] \ne Q[x_2, y_2]
  \implies 
  \lVert g(x_1, y_1) - g(x_2, y_2) \rVert \ge \varepsilon
  \ ,
\end{align*}
where \(g(x, y) = (z(x, y), p(x, y))\).

\begin{theorem*}[\ref{thm:main2}]
  Suppose \(p\) is connected, discrete and \(\varepsilon\)-separable,
  while \(z\) is continuous.
  If \(\{(x_i, y_i)\}_{i=1}^\infty\) is a dense sequence in \([0, 1]^2\),
  then there exists an algorithm which inputs
  \(z(x_i, y_i)\) and \(p(x_i, y_i)\)
  for \(i = 1,\dots,n\),
  and correctly determines,
  for \(k = 1,\dots, n\),
  the sets
  \begin{align*}
    \mathcal{I}_n(k) = \{ 1 \le i \le n \mid Q[x_i, y_i] = Q[x_k, y_k] \}
    \ ,
  \end{align*}
  provided \(n\) is large enough and \(\varepsilon\) is known (or lower-bounded).
\end{theorem*}

In essence, \cref{thm:main2} claims that
any profile \(p\) satisfying certain assumptions
makes it possible to algorithmically determine
inherent radiance segmentation for any associated image,
provided the camera resolution is sufficiently high.
The assumption of \(p\) being connected intuitively means that
any uniformly colored region \(\mathcal{H}(x, y)\)
in the scene with the water removed appears as a single connected patch.
The assumption of \(p\) being discrete
ensures that each \(\mathcal{H}(x, y)\) is sampled an infinite number of times in the limit \(n \to \infty\),
while \(p\) being \(\varepsilon\)-separable strongly forbids metamerism.
Finally, the assumption of \(z\) being continuous avoids many mathematically pathological scenarios
of little practical relevance,
and generalizing to piecewise continuous \(z\) seems possible.

\subsection{Reader's guide}

In \cref{sec:imf} we outline a derivation of \eqref{eq:imf} for completeness.
In \cref{sec:ill-posedness}, we describe the solution space for color restoration in the fully general setting,
formally verifying that it is an ill-posed problem.
In \cref{sec:att-funcs}, we investigate a class of functions \(\mathbb{R} \to \mathbb{R}\) that we call \emph{attenuating},
and we prove certain extrapolation properties on these.
In \cref{sec:using-irs}, we show how inherent radiance segmentation
reduces color restoration to extrapolation of said functions,
proving \cref{thm:main1-eval}.
In \cref{sec:obtaining-irs}, we show how to obtain inherent radiance segmentation in an idealized setting,
proving \cref{thm:main2}.
Finally, we give a few concluding remarks in \cref{sec:conclusion}.

\section{Image formation}
\label{sec:imf}

Here, we briefly explain the derivation of \eqref{eq:imf},
where an observer is looking at a point in the scene
along a horizontal line of sight of length \(z\).
A more general derivation may be found in \cite{preisendorfer1976hydrologic}.
For any wavelength \(\lambda\),
we denote the apparent radiance at distance \(z\) by \(F(\lambda, z)\),
so that \(F(\lambda, 0) = L(\lambda)\) is the inherent radiance.
In the radiative transfer equation,
the instantaneous radiance loss at \(z\)
due to absorption and scattering is \(c(\lambda)F(\lambda, z)\),
where \(c\) is the beam attenuation coefficient.
The instantaneous radiance gain due to in-scattering from other directions,
as well as potential emission,
we denote by \(F_+(\lambda, z)\),
so that
\begin{align*}
  \frac{\partial F}{\partial z} = -cF + F_+
  \ .
\end{align*}
Using the integrating factor \(e^{cz}\),
the general solution to this linear equation is found to be
\begin{align*}
  F(\lambda, z) = e^{-c(\lambda)z}F(\lambda, 0) + \textstyle\int_0^ze^{-c(\lambda)(z - \tilde{z})}F_+(\lambda, \tilde{z})\mathrm{d}\tilde{z}
  \ .
\end{align*}
In general, \(F_+\) depends on the geometry of the scene and can be very complicated;
however,
we employ the common simplifying assumption that it is constant in \(z\), which yields
\begin{align*}
  F(\lambda, z)
  &= e^{-c(\lambda)z}F(\lambda, 0) + (1 - e^{-c(\lambda)z})\frac{F_+(\lambda)}{c(\lambda)} \\
  &= e^{-c(\lambda)z}L(\lambda) + (1 - e^{-c(\lambda)z})B(\lambda)
  \ ,
\end{align*}
where \(B(\lambda) = F_+(\lambda)/c(\lambda)\); this assumption requires a horizontal view.
If the signal \(F(\lambda, z)\) is captured by a camera with spectral sensitivities \(S_j\),
where \(j\) is the channel index,
the resulting pixel intensities become
\begin{align*}
  P_{i,j}(z)
  = \textstyle\int_\Lambda (e^{-cz}L_i + (1 - e^{-cz})B) S_j \mathrm{d}\lambda
  \ ,
\end{align*}
where \(L_i\) is the inherent radiance associated with pixel \(i\).

\section{Ill-posedness}
\label{sec:ill-posedness}

In this section,
we explain why
\begin{align*}
  P_{i,j}(z) = \textstyle\int_\Lambda(e^{-cz}L_i + (1 - e^{-cz})B)S_j\mathrm{d}\lambda
\end{align*}
cannot be estimated from the \(P_{i,j}(z_i)\) and the \(z_i\) alone,
unless \(z = z_i\) or \(cS_j \equiv 0\).
In fact, this is true even when the \(S_j\) are known,
hence we assume here that the only unknowns are \(c,B\) and \(L_i\).
The following definition will allow us to formally reason about solution spaces:

\begin{definition}
  For any \(z \in \mathbb{R}\) and \(c, B, L: \Lambda \to \mathbb{R}\), define
  \begin{align*}
    F^{(z)}(c, B, L)
    &= e^{-cz}L + (1 - e^{-cz})B \ , \\
    P_j^{(z)}(c, B, L)
    &= \textstyle\int_\Lambda F^{(z)}(c, B, L)S_j\mathrm{d}\lambda \ , \\
    \mathcal{E}^{(z)}(c, B, L)
    &= \{ (\hat{c}, \hat{B}, \hat{L}) \mid P_j^{(z)}(\hat{c}, \hat{B}, \hat{L}) 
    = P_j^{(z)}(c, B, L) \ \forall j\}
      .
  \end{align*}
  For any \(z_1,z_2 \in \mathbb{R}\), we might write
  \begin{align*}
    &P_j^{(z_1)}(\mathcal{E}^{(z_2)}(c, B, L)) \\
    &= \{ P_j^{(z_1)}(\hat{c}, \hat{B}, \hat{L}) \mid (\hat{c}, \hat{B}, \hat{L}) \in \mathcal{E}^{(z_2)}(c, B, L)\}
    \ .
  \end{align*}
\end{definition}
Intuitively, \(\mathcal{E}^{(z)}(c, B, L)\) is the set of all tuples \((\hat{c}, \hat{B}, \hat{L})\)
that could have produced the pixel intensities \(P_j^{(z)}(c, B, L)\).
In the following lemma, we give an explicit description of this set:

\begin{lemma}
  \label{lem:desc-sol-space}
  For any \(z \in \mathbb{R}\) and \(c, B, L: \Lambda \to \mathbb{R}\), we have \( \mathcal{E}^{(z)}(c, B, L) = \mathcal{E}_0^{(z)}(c, B, L)\),
  where \(\mathcal{E}_0^{(z)}(c, B, L)\) is
  \begin{align*}
    \{ (\hat{c}, \hat{B}, \hat{L}) \ | \
      &\hat{L} = e^{\hat{c}z}F^{(z)}(c,B,L) + (1 - e^{\hat{c}z})\hat{B} + L^\perp \ , \\
      &\text{where } \textstyle\int_\Lambda e^{-\hat{c}z}L^\perp S_j \mathrm{d}\lambda = 0
      \quad \forall j\}
      \ .
  \end{align*}
\end{lemma}
\begin{proof}
  If \((\hat{c}, \hat{B}, \hat{L}) \in \mathcal{E}_0^{(z)}(c, B, L)\), then
  \begin{align*}
    \hat{L}
    &= e^{\hat{c}z}F^{(z)}(c,B,L) + (1 - e^{\hat{c}z})\hat{B} + L^\perp \ , \\
    P_j^{(z)}(\hat{c}, \hat{B}, \hat{L})
    &= \textstyle\int_\Lambda F^{(z)}(\hat{c}, \hat{B}, \hat{L})S_j\mathrm{d}\lambda \\
    &= \textstyle\int_\Lambda (e^{-\hat{c}z}\hat{L} + (1 - e^{-\hat{c}z})\hat{B})S_j\mathrm{d}\lambda \\
    &= \textstyle\int_\Lambda (F^{(z)}(c, B, L) + e^{-\hat{c}z}L^\perp)S_j\mathrm{d}\lambda \\
    &= P_j^{(z)}(c, B, L)
      \ ,
  \end{align*}
  hence \(\mathcal{E}_0^{(z)}(c, B, L) \subseteq \mathcal{E}^{(z)}(c, B, L)\).

  On the other hand, if \((\hat{c}, \hat{B}, \hat{L}) \in \mathcal{E}^{(z)}(c, B, L)\), then
  \begin{align*}
    P_j^{(z)}(\hat{c}, \hat{B}, \hat{L})
    &= P_j^{(z)}(c, B, L) \ , \\
    \textstyle\int_\Lambda F^{(z)}(\hat{c}, \hat{B}, \hat{L})S_j\mathrm{d}\lambda
    &=
      \textstyle\int_\Lambda F^{(z)}(c, B, L)S_j\mathrm{d}\lambda \ , \\
    F^{(z)}(\hat{c}, \hat{B}, \hat{L})
    &= F^{(z)}(c, B, L) + F^\perp
  \end{align*}
  for some \(F^\perp\) satisfying \(\textstyle\int_\Lambda F^\perp S_j\mathrm{d}\lambda = 0\) for all \(j\).
  Letting \(L^\perp = e^{\hat{c}z}F^\perp\), we have
  \begin{align*}
    e^{-\hat{c}z}\hat{L} + (1 - e^{-\hat{c}z})\hat{B}
    & = F^{(z)}(c, B, L) + e^{-\hat{c}z}L^\perp \ , \\
    \hat{L}
    &= e^{\hat{c}z}F^{(z)}(c, B, L) + (1 - e^{\hat{c}z})\hat{B} + L^\perp \ , \\
      \textstyle\int_\Lambda e^{-\hat{c}z}L^\perp S_j \mathrm{d}\lambda &= \textstyle\int_\Lambda F^\perp S_j \mathrm{d}\lambda = 0 \ ,
  \end{align*}
  hence \(\mathcal{E}^{(z)}(c, B, L) \subseteq \mathcal{E}_0^{(z)}(c, B, L)\).
\end{proof}

The fact that all members of \(\mathcal{E}^{(z)}(c, B, L)\) yield identical pixel intensities
is captured in the following trivial lemma
whose sole purpose is to be contrasted with \cref{thm:ill-posedness} that follows it.
\begin{lemma}
  For any \(z \in \mathbb{R}\) and \(c, B, L: \Lambda \to \mathbb{R}\), we have
  \(P_j^{(z)}(\mathcal{E}^{(z)}(c, B, L)) = \{P_j^{(z)}(c, B, L)\}\).
\end{lemma}
\begin{proof}
  \begin{align*}
    &P_j^{(z)}(\mathcal{E}^{(z)}(c, B, L)) \\
    &\doteq \{ P_j^{(z)}(\hat{c}, \hat{B}, \hat{L}) \mid (\hat{c}, \hat{B}, \hat{L}) \in \mathcal{E}^{(z)}(c, B, L)\} \\
    &= \{P_j^{(z)}(\hat{c}, \hat{B}, \hat{L}) \mid P_j^{(z)}(\hat{c}, \hat{B}, \hat{L}) = P_j^{(z)}(c, B, L)\} \\
    &= \{P_j^{(z)}(c, B, L)\} \ .      
  \end{align*}
\end{proof}
Given only the pixel value \(P_j^{(z_1)}(c, B, L)\) and the associated distance \(z_1\),
by definition, we cannot know anything about the tuple \((c, B, L)\) besides its membership in \(\mathcal{E}^{(z_1)}(c, B, L)\).
If we then wish to estimate \(P_j^{(z_2)}(c, B, L)\) at some other distance \(z_2 \ne z_1\),
the possible candidates are exactly the elements of \(P_j^{(z_2)}(\mathcal{E}^{(z_1)}(c, B, L))\).
Unfortunately, this set contains all real numbers:
\begin{theorem}
  \label{thm:ill-posedness}
  For any \(z_1,z_2 \in \mathbb{R}\) and \(c, B, L: \Lambda \to \mathbb{R}\), we have
  \(P_j^{(z_2)}(\mathcal{E}^{(z_1)}(c, B, L)) = \mathbb{R}\)
  if \(z_1 \ne z_2\) and \(cS_j \not\equiv 0\).
\end{theorem}
\begin{proof}
  For any \(\hat{B}:\Lambda \to \mathbb{R}\), let
  \begin{align*}
    \hat{L} = e^{cz_1}F^{(z_1)}(c,B,L) + (1 - e^{cz_1})\hat{B} \ .
  \end{align*}
  \cref{lem:desc-sol-space} guarantees that \((c, \hat{B}, \hat{L}) \in \mathcal{E}^{(z_1)}(c, B, L)\).
  Consequently, \(P_j^{(z_2)}(\mathcal{E}^{(z_1)}(c, B, L))\) contains
  \begin{align*}
    &P_j^{(z_2)}(c, \hat{B}, \hat{L}) \\
    &= \textstyle\int_\Lambda F^{(z_2)}(c, \hat{B}, \hat{L}) S_j \mathrm{d}\lambda \\
    &= \textstyle\int_\Lambda(e^{-cz_2}\hat{L} + (1 - e^{-cz_2})\hat{B}) S_j \mathrm{d}\lambda \\
    &= \textstyle\int_\Lambda(e^{c(z_1 - z_2)}F^{(z_1)}(c, B, L) \\
      &\quad + e^{-cz_2}(1 - e^{cz_1})\hat{B}
      + (1 - e^{-cz_2})\hat{B}) S_j \mathrm{d}\lambda \\
    &= \textstyle\int_\Lambda e^{c(z_1 - z_2)}F^{(z_1)}(c, B, L) S_j \mathrm{d}\lambda
      + \textstyle\int_\Lambda(1 - e^{c(z_1 - z_2)})\hat{B} S_j \mathrm{d}\lambda
      \ .
  \end{align*}
  The first term in the above sum is completely determined
  while the last depends on the arbitrary \(\hat{B}\).
  Since \(z_1 \ne z_2\) and \(cS_j \not\equiv 0\),
  there exists some \(\hat{B}\) for which the last term does not vanish.
  Moreover, by freely scaling this \(\hat{B}\),
  we can force \(P_j^{(z_2)}(c, \hat{B}, \hat{L})\) to take on any real value,
  hence \(P_j^{(z_2)}(\mathcal{E}^{(z_1)}(c, B, L)) = \mathbb{R}\).
\end{proof}

In conclusion,
knowing only the pixel intensities \(P_{i,j}(z_i)\) and the distances \(z_i\)
gives no information about \(P_{i,j}(z)\) unless \(z = z_i\) or \(cS_j \equiv 0\).

\section{Attenuating functions}
\label{sec:att-funcs}

In \cref{sec:ill-posedness},
we verified that the task of estimating \(P_{i,j}(z)\)
from \(P_{i,j}(z_i)\) and \(z_i\) is impossible without additional information.
In \cref{sec:using-irs},
we shall see that knowledge of inherent radiance segmentation
reduces this task to extrapolation of attenuating functions.
In the current section,
we explain how such extrapolation can be carried out
as a linear combination of known quantities.
\begin{definition}
  \label{def:att-func}
  A function \(P: [0, \infty) \to \mathbb{R}\) is \emph{attenuating} if it can be written as
  \begin{align*}
    P(z)
    = \textstyle\int_\Lambda e^{qz}Q \mathrm{d}\lambda
    \ ,
  \end{align*}
  where \(q,Q:\Lambda \to \mathbb{R}\) on some bounded interval \(\Lambda \subset \mathbb{R}\).
  As shorthand notation, we denote this by \(P = P[q, Q]\).
\end{definition}
We treat \(q\) and \(Q\) as hidden parameters defining \(P\),
however, we will assume that some upper and lower bounds on \(q\) are known across \(\Lambda\).
In the context of underwater color restoration,
this corresponds to knowing upper and lower bounds on
the beam attenuation coefficient \(c\),
which is reasonable due to Jerlov water types being well-studied \cite{jerlov1976marine}.
Once such bounds are known,
then the possible candidates for \(P(z)\) become restricted to a finite-length interval,
provided some evaluations or derivatives of \(P\) are known at some \(z_i \ne z\).
The more of such ``information elsewhere'' is available,
the smaller this uncertainty interval becomes before vanishing in the limit.
Evaluations are discussed in \cref{subsec:evaluations}
and are of most relevance to underwater color restoration,
while derivatives are discussed in \cref{subsec:derivatives}
and are included for completeness.

\subsection{Extrapolation from evaluations}
\label{subsec:evaluations}
We begin by defining some auxiliary quantities
useful for describing the error for extrapolation from evaluations.
\begin{definition}
  \label{def:delta-sample}
  For any \(\mathcal{Z} = \{z_i\}_{i=1}^n \subset [0, \infty)\)
  and any \(\mathcal{A} = \{\alpha_i\}_{i=1}^n\)
  with \(\alpha_i:[0, \infty) \to \mathbb{R}\),
  define
  \begin{align*}
    \Psi_{\mathcal{Z},\mathcal{A}}(z, u) = \textstyle\sum_{i=1}^n\alpha_i(z)u^{z_i} - u^z
    \ ,
  \end{align*}
  where \(u\) is a non-negative real-valued variable.
  Furthermore, for any bounds \(u_1,u_2 \in \mathbb{R}\) with \(0 \le u_1 < u_2\), define
  \begin{align*}
    \Psi_{\mathcal{Z},\mathcal{A}}^{u_1,u_2}(z)
    = \max_{u_1 \le u \le u_2} |\Psi_{\mathcal{Z},\mathcal{A}}(z, u)|
    \ .
  \end{align*}
\end{definition}
As we shall see,
any choice of the coefficients \(\alpha_i(z)\)
yields the approximation \(P(z) \approx \sum_{i=1}^n\alpha_i(z)P(z_i)\)
with a bounded error,
though making these bounds tight
requires minimizing \(\Psi_{\mathcal{Z},\mathcal{A}}^{u_1,u_2}(z)\)
with respect to \(\mathcal{A} = \{\alpha_i\}_{i=1}^n\) at \(z\).
\begin{lemma}
  \label{lem:est-by-eval}
  Let \(P = P[q, Q]\),
  \(\mathcal{Z} = \{z_i\}_{i=1}^n \subset [0, \infty)\)
  and \(\mathcal{A} = \{\alpha_i\}_{i=1}^n\) with \(\alpha_i: [0, \infty) \to \mathbb{R}\).
  Then
  \begin{align*}
    \textstyle\sum_{i=1}^n\alpha_i(z)P(z_i)
    &= P(z) + \textstyle\int_\Lambda \Psi_{\mathcal{Z},\mathcal{A}}(z, e^{q})Q \mathrm{d}\lambda
      \ .
  \end{align*}
\end{lemma}
\begin{proof}
  \begin{align*}
    \textstyle\sum_i\alpha_i(z)P(z_i)
    &= \textstyle\sum_i\alpha_i(z) \textstyle\int_\Lambda e^{qz_i}Q \mathrm{d}\lambda \\
    &= \textstyle\int_\Lambda \textstyle\sum_i\alpha_i(z)e^{qz_i}Q \mathrm{d}\lambda \\
    &= \textstyle\int_\Lambda (e^{qz} + \Psi_{\mathcal{Z},\mathcal{A}}(z, e^{q}))Q \mathrm{d}\lambda \\
    &= \textstyle\int_\Lambda e^{qz}Q \mathrm{d}\lambda
      + \textstyle\int_\Lambda \Psi_{\mathcal{Z},\mathcal{A}}(z, e^{q})Q \mathrm{d}\lambda \\
    &= P(z) + \textstyle\int_\Lambda \Psi_{\mathcal{Z},\mathcal{A}}(z, e^{q})Q \mathrm{d}\lambda
      \ .
  \end{align*}
\end{proof}
\cref{lem:est-by-eval} isolates the error term for the extrapolation,
but it does not immediately allow us to bound it.
For this, we need upper and lower bounds on \(e^{q(\lambda)}\) for all \(\lambda \in \Lambda\).
\begin{lemma}
  \label{lem:bound-by-eval}
  Let \(P = P[q, Q]\) and \(u_1,u_2 \in \mathbb{R}\)
  such that \(u_1 \le e^{q(\lambda)} \le u_2\) for all \(\lambda \in \Lambda\).
  Then
  \begin{align*}
    |\textstyle\int_\Lambda \Psi_{\mathcal{Z},\mathcal{A}}(z, e^{q})Q \mathrm{d}\lambda|
    \le \Psi_{\mathcal{Z},\mathcal{A}}^{u_1,u_2}(z)\textstyle\int_\Lambda |Q| \mathrm{d}\lambda
    \ .
  \end{align*}
\end{lemma}
\begin{proof}
  For every \(\lambda \in \Lambda\), we have
  \begin{align*}
    |\Psi_{\mathcal{Z},\mathcal{A}}(z, e^{q(\lambda)})|
    \le \max_{u_1 \le u \le u_2} |\Psi_{\mathcal{Z},\mathcal{A}}(z, u)| 
    = \Psi_{\mathcal{Z},\mathcal{A}}^{u_1,u_2}(z)
      \ .
  \end{align*}
  Therefore,
  \begin{align*}
    |\textstyle\int_\Lambda \Psi_{\mathcal{Z},\mathcal{A}}(z, e^{q})Q \mathrm{d}\lambda|
    &\le \textstyle\int_\Lambda |\Psi_{\mathcal{Z},\mathcal{A}}(z, e^{q}) Q| \mathrm{d}\lambda \\
    &= \textstyle\int_\Lambda |\Psi_{\mathcal{Z},\mathcal{A}}(z, e^{q})| |Q| \mathrm{d}\lambda \\
    &\le \textstyle\int_\Lambda \Psi_{\mathcal{Z},\mathcal{A}}^{u_1,u_2}(z) |Q| \mathrm{d}\lambda \\
    &= \Psi_{\mathcal{Z},\mathcal{A}}^{u_1,u_2}(z)\textstyle\int_\Lambda |Q| \mathrm{d}\lambda \ .
  \end{align*}
\end{proof}
Combining \cref{lem:est-by-eval} and \cref{lem:bound-by-eval} into a single statement
yields our first main result:
\begin{corollary}
  \label{cor:main1-eval}
  Let \(P = P[q, Q]\),
  \(\mathcal{Z} = \{z_i\}_{i=1}^n \subset [0, \infty)\),
  and let \(\mathcal{A} = \{\alpha_i\}_{i=1}^n\) with \(\alpha_i: [0, \infty) \to \mathbb{R}\).
  Furthermore, let \(u_1, u_2 \in \mathbb{R}\) such that \(u_1 \le e^{q(\lambda)} \le u_2\) for all \(\lambda \in \Lambda\).
  Then
  \begin{align*}
    \textstyle\sum_i\alpha_i(z)P(z_i)
    &= P(z) \pm \Psi_{\mathcal{Z},\mathcal{A}}^{u_1,u_2}(z)\textstyle\int_\Lambda |Q| \mathrm{d}\lambda
      \ ,
  \end{align*}
  where the notation \(y = x \pm \delta\) means \(y \in \{x + w\mid |w| \le |\delta|\}\).
\end{corollary}
For an illustration of \cref{cor:main1-eval},
see \cref{fig:extrapolation} and \cref{fig:alphas}.
The practical utility comes from the fact that the error bound for extrapolation
contains the factor
\begin{align*}
  \Psi_{\mathcal{Z},\mathcal{A}}^{u_1,u_2}(z)
  &= \max_{u_1 \le u \le u_2} |\textstyle\sum_{i=1}^n\alpha_i(z)u^{z_i} - u^z|
    \ .
\end{align*}
The extent to which this quantity can be minimized
with respect to the coefficients \(\alpha_i(z)\)
depends on how well the continuous function \(u^z\)
can be approximated by a linear combination of the functions \(u^{z_i}\)
on the interval \([u_1, u_2]\).
This is a standard question in approximation theory
and admits an elegant answer in the limit \(n \to \infty\).
In the following, let \(C[0, 1]\) denote
the space of continuous functions \([0, 1] \to \mathbb{R}\).
\begin{theorem*}[Müntz–Szász]
  Let \(\{z_i\}_{i=0}^\infty \subset [0, \infty)\) with \(z_0 = 0\)
  be an increasing sequence,
  and let \(u\) be a variable ranging over \([0, 1]\).
  Then the \(\mathbb{R}\)-span of \(\{u^{z_i}\}_{i=0}^\infty\) is dense in \(C[0, 1]\) under the supremum norm
  if and only if \(\sum_{i=1}^\infty z_i^{-1} = \infty\).
\end{theorem*}
For reference, see e.g. \cite{almira2007muntz}.
Adapting this result to our setting requires two things:
a change of variables to go from \([0, 1]\) to \([u_1, u_2]\)
and removing the requirement of having \(z_0 = 0\).
The reason for caring about the latter
is that it makes little sense to assume knowledge of \(P(0)\) in color restoration,
where \(P(0)\) is the quantity we will be most interested in estimating in the first place.
\begin{theorem}
  \label{thm:limit-eval}
  Let \(P = P[q, Q]\) and \(u_1, u_2 \in \mathbb{R}\) such that \(0 < u_1 \le e^{q(\lambda)} \le u_2\) for all \(\lambda \in \Lambda\).
  Let \(\{z_i\}_{i=0}^\infty \subset [0, \infty)\) be an increasing sequence satisfying \(\sum_{i = 1}^\infty z_i^{-1} = \infty\).
  Finally, let \(\mathcal{Z}(n) = \{z_i\}_{i=0}^n\), and for any \(z \in [0, \infty)\) let
  \begin{align*}
    \mathcal{A}(n) = \{\alpha_{n, i}\}_{i=0}^n = \arg \min_{\mathcal{A}(n)} \Psi_{\mathcal{Z}(n), \mathcal{A}(n)}^{u_1,u_2}(z)
    \ .
  \end{align*}
  Then
  \begin{align*}
    \lim_{n \to \infty} \textstyle\sum_{i=0}^n\alpha_{n, i}(z)P(z_i) = P(z)
    \ .
  \end{align*}
\end{theorem}
\begin{proof}
  Due to \cref{cor:main1-eval},
  it suffices to show that
  \begin{align*}
    \lim_{n \to \infty} \Psi_{\mathcal{Z}(n),\mathcal{A}(n)}^{u_1,u_2}(z) = 0 \ .
  \end{align*}
  Consider the variable change \(\mu = u/u_2\), and let
  \begin{align*}
    f_z(\mu) =
    \begin{cases}
      (u_2\mu)^z &\text{ if } \mu \ge u_1/u_2 \\
      u_1^{z-1}u_2 \mu &\text{ if } \mu < u_1/u_2 \\
    \end{cases}
    \ .
  \end{align*}
  Observe that \(f_z \in C[0, 1]\).
  Since \(0, z_1, z_2, \dots\) is an increasing sequence
  satisfying \(\sum_{i=1}^\infty z_i^{-1} = \infty\), then
  \begin{align*}
    \lim_{n\to\infty} \max_{0 \le \mu \le 1} |\gamma_{n,0}(z) + \textstyle\sum_{i=1}^n \gamma_{n, i}(z)\mu^{z_i} - f_z(\mu)| = 0
  \end{align*}
  for some coefficients \(\gamma_{n, i}(z) \in \mathbb{R}\)
  due to Müntz–Szász.
  However, since \(f_z(0) = 0\) and \(z_i > 0\) for \(i \ge 1\), then
  \begin{align*}
    \lim_{n \to \infty} \gamma_{n, 0}(z) = 0
    \ .
  \end{align*}
  Consequently, we may omit the term for \(i = 0\) and still have
  \begin{align*}
    \lim_{n \to \infty} \max_{0 \le \mu \le 1} |\textstyle\sum_{i=1}^n \gamma_{n, i}(z)\mu^{z_i} - f_z(\mu)| = 0
    \ .
  \end{align*}
  Let \(\tilde{\mathcal{A}}(n) = \{\tilde{\alpha}_{n,i}\}_{i=0}^n\) with \(\tilde{\alpha}_{n,0}(z) = 0\) and
  \begin{align*}
    \tilde{\alpha}_{n,i}(z) = \gamma_{n,i}(z)u_2^{-z_i} \text{ for } i \ge 1 \ .
  \end{align*}
  Since \(f_z(\mu) = u_2^z\mu^z\) for \(\mu \ge u_1/u_2\), then
  \begin{align*}
    \Psi_{\mathcal{Z}(n),\tilde{\mathcal{A}}(n)}^{u_1,u_2}(z)
    &= \max_{u_1/u_2 \le \mu \le 1} |\textstyle\sum_{i=1}^n \gamma_{n,i}(z)\mu^{z_i} - u_2^z\mu^z| \\
    &= \max_{u_1/u_2 \le \mu \le 1} |\textstyle\sum_{i=1}^n \gamma_{n,i}(z)\mu^{z_i} - f_z(\mu)| \\
    &\le \max_{0 \le \mu \le 1} |\textstyle\sum_{i=1}^n \gamma_{n,i}(z)\mu^{z_i} - f_z(\mu)|
      \mathrel{\mathop{\to}\limits_{n\to\infty}} 0 \ .
  \end{align*}
  Since the \(\alpha_{n,i}\) minimize \(\Psi_{\mathcal{Z}(n),\mathcal{A}(n)}^{u_1,u_2}(z)\), then
  \begin{align*}
    0 \le \Psi_{\mathcal{Z}(n),\mathcal{A}(n)}^{u_1,u_2}(z)
    \le \Psi_{\mathcal{Z}(n),\tilde{\mathcal{A}}(n)}^{u_1,u_2}(z)
    \ ,
  \end{align*}
  and it follows from the squeeze theorem that
  \begin{align*}
    \lim_{n\to\infty}\Psi_{\mathcal{Z}(n),\mathcal{A}(n)}^{u_1,u_2}(z) = 0 \ .
  \end{align*}
\end{proof}

\begin{figure}[htb]
  \centering
  \includegraphics[width=1.0\linewidth]{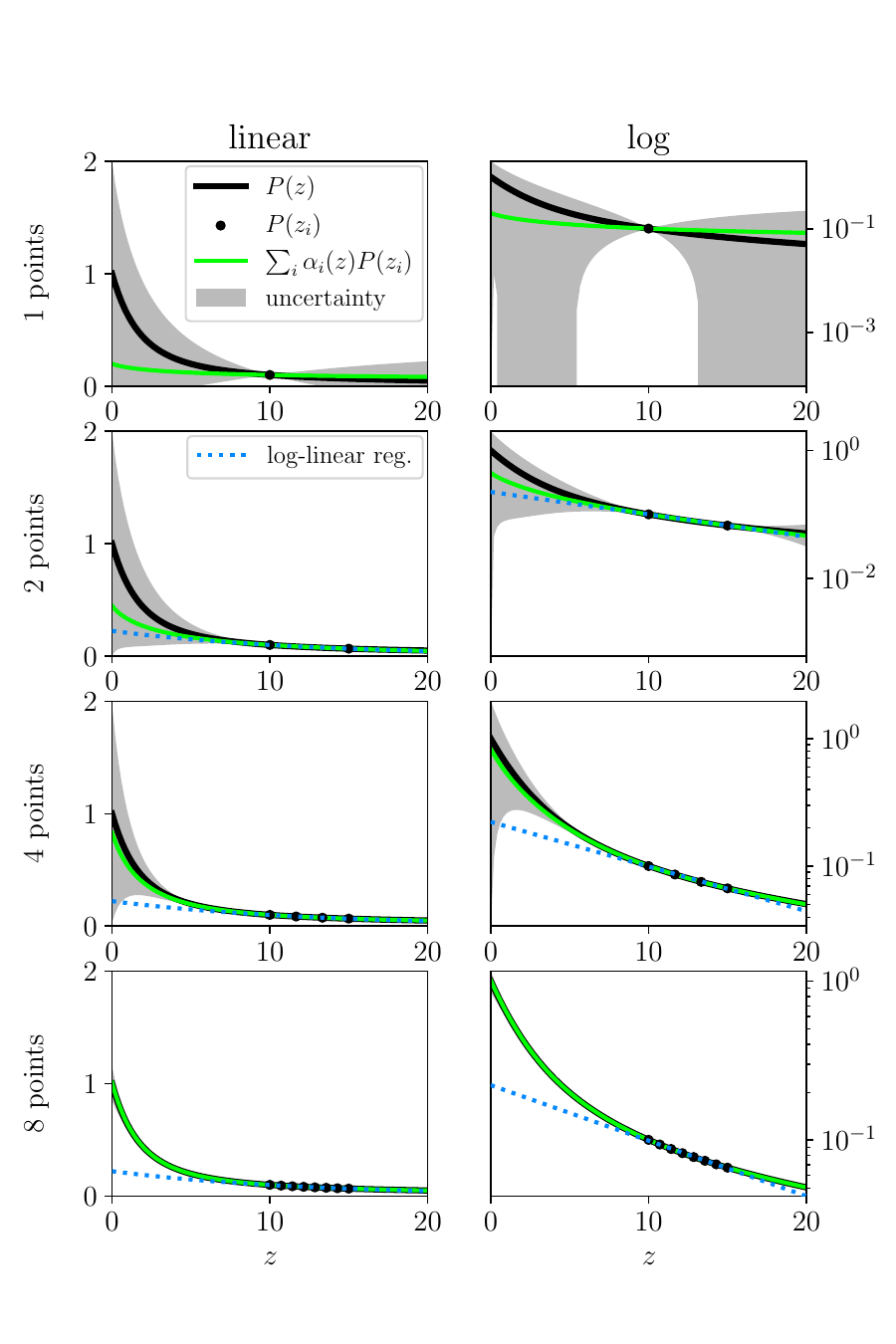}
  \caption{
    Extrapolation of \(P = P[q, Q]\), where \(q(\lambda) = -\lambda\), \(Q(\lambda) = 1\)
    and \(\Lambda = [0, 1]\).
    The black curve describes \(P(z)\), which is treated as unknown
    except for the points \(z_i\) (black dots) that are linearly spaced between \(z=10\) and \(z=15\).
    The green curve describes the extrapolation \(\sum_i\alpha_i(z)P(z_i) \approx P(z)\),
    where the \(\alpha_i\) were optimized numerically
    using \(u_1= 1/e\) and \(u_2 = 1\).
    The gray region describes the uncertainty bounds
    as given by \cref{lem:bound-by-eval}.
    Dotted blue denotes log-linear regression for comparison.
  }
  \label{fig:extrapolation}
\end{figure}

\begin{figure}[htb]
  \centering
  \includegraphics[width=1.0\linewidth]{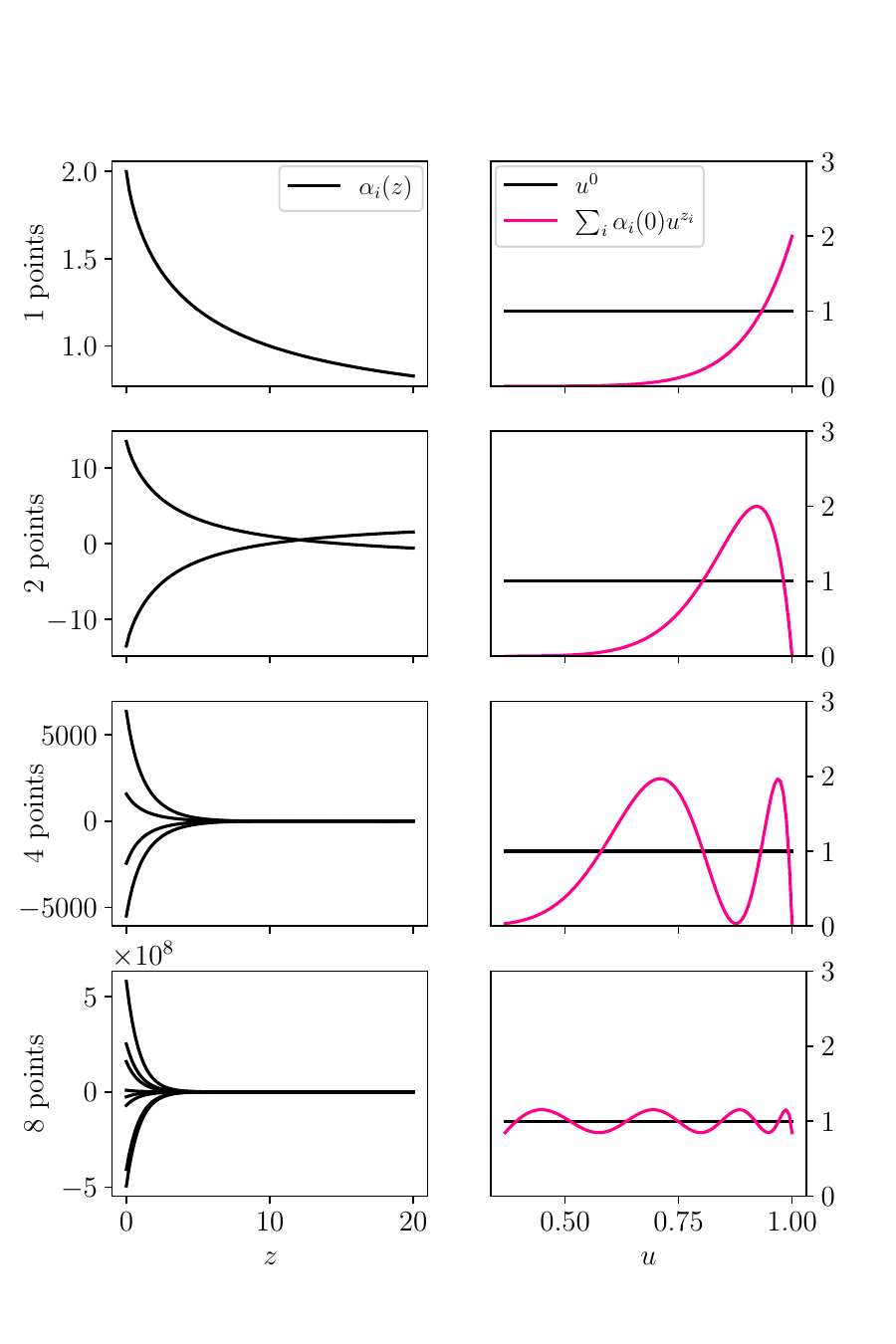}
  \caption{
    With the context of \cref{fig:extrapolation},
    the black curves in the left column describe the coefficients \(\alpha_i(z)\).
    In the right column,
    magenta is \(\sum_i\alpha_i(0)u^{z_i}\) for \(u_1 \le u \le u_2\)
    and black is \(u^0 = 1\),
    while their difference is \(\Psi_{\mathcal{Z}, \mathcal{A}}(0, u)\).
  }
  \label{fig:alphas}
\end{figure}
Intuitively, \cref{thm:limit-eval} claims that if \(P = P[q, Q]\), then
\(P(z)\) can be estimated to arbitrary precision
from the \(z_i\) and the \(P(z_i)\),
provided \(q\) has known upper and lower bounds,
\(n\) is large enough, and \(z_n\) does not grow too fast as \(n \to \infty\).
\subsection{Extrapolation from derivatives}
\label{subsec:derivatives}
Extrapolation of attenuating functions from derivatives
is completely analogous to extrapolation from evaluations.
The main difference is that the evaluations \(P(z_i)\) at multiple points \(z_i\)
are replaced with multiple derivatives \(\frac{\mathrm{d}^iP}{\mathrm{d}z^i}(z_0)\) at a single point \(z_0\).
Since images are, in practice, represented by evaluations,
this perspective appears to be of lesser relevance to color restoration.
Still, we include this if only to highlight the mathematical structure of attenuating functions.
That being said, derivatives can be estimated from evaluations,
and it is conceivable that
combining the two kinds of extrapolation\textemdash using evaluations and derivatives\textemdash
might yield tighter error bounds.
We leave such considerations for future research,
focusing our current scope on asymptotic behavior instead.

Let \(\mathbb{N} = \{0, 1, 2, \dots\}\) be the set of non-negative integers.
The following definition corresponds to \cref{def:delta-sample}:
\begin{definition}
  \label{def:delta-diff}
  For any \(z_0 \in \mathbb{R}\),
  any finite \(\mathcal{D} \subset \mathbb{N}\)
  and any \(\mathcal{B} = \{\beta_i\}_{i \in \mathcal{D}}\) with \(\beta_i:[0, \infty) \to \mathbb{R}\), define
  \begin{align*}
    \Phi_{\mathcal{D}, \mathcal{B}}(z_0, z, v)
    = \textstyle\sum_{i \in \mathcal{D}}\beta_i(z)v^i - e^{v(z - z_0)}
    \ ,
  \end{align*}
  where \(v\) is a real-valued variable.
  Furthermore, for any bounds \(v_1,v_2 \in \mathbb{R}\) with \(v_1 < v_2\), define
  \begin{align*}
    \Phi_{\mathcal{D}, \mathcal{B}}^{v_1,v_2}(z_0, z)
    = \max_{v_1 \le v \le v_2} |\Phi_{\mathcal{D}, \mathcal{B}}(z_0, z, v)|
    \ .
  \end{align*}
\end{definition}
As for evaluations,
any choice of the coefficients \(\beta_i(z)\)
yields the approximation \(P(z) \approx \sum_{i=1}^n\beta_i(z)\frac{\mathrm{d}^iP}{\mathrm{d}z^i}(z_0)\)
with a bounded error,
though making these bounds tight
requires minimizing \(\Phi_{\mathcal{D}, \mathcal{B}}^{v_1,v_2}(z_0, z)\)
with respect to \(\mathcal{B} = \{\beta_i\}_{i \in \mathcal{D}}\) at \(z\).
The derivative version of \cref{lem:est-by-eval} is as follows:
\begin{lemma}
  \label{lem:est-by-diff}
  Let \(P = P[q, Q]\),
  \(z_0 \in \mathbb{R}\),
  \(\mathcal{D} \subset \mathbb{N}\) finite
  and \(\mathcal{B} = \{\beta_i\}_{i \in \mathcal{D}}\) with \(\beta_i:[0, \infty) \to \mathbb{R}\).
  Then
  \begin{align*}
    \textstyle\sum_{i \in \mathcal{D}}\beta_i(z)\frac{\mathrm{d}^iP}{\mathrm{d}z^i}(z_0)
    &= P(z) + \textstyle\int_\Lambda \Phi_{\mathcal{D},\mathcal{B}}(z_0,z,q)e^{qz_0}Q \mathrm{d}\lambda
      \ .
  \end{align*}
\end{lemma}
\begin{proof}
  Since
  \begin{align*}
    \textstyle\frac{\mathrm{d}^iP}{\mathrm{d}z^i}(z)
    &= \textstyle\frac{\mathrm{d}^i}{\mathrm{d}z^i} \textstyle\int_\Lambda e^{q(\lambda)z}Q \mathrm{d}\lambda \\
    &= \textstyle\int_\Lambda \frac{\mathrm{d}^i}{\mathrm{d}z^i}e^{qz}Q \mathrm{d}\lambda \\
    &= \textstyle\int_\Lambda q^ie^{qz}Q \mathrm{d}\lambda
      \ ,
  \end{align*}
  then
  \begin{align*}
    &\textstyle\sum_{i \in \mathcal{D}}\beta_i(z)\frac{\mathrm{d}^iP}{\mathrm{d}z^i}(z_0) \\
    &= \textstyle\sum_{i \in \mathcal{D}}\beta_i(z) \textstyle\int_\Lambda q^ie^{qz_0}Q \mathrm{d}\lambda \\
    &= \textstyle\int_\Lambda (\textstyle\sum_{i \in \mathcal{D}}\beta_i(z)q^i)e^{qz_0}Q \mathrm{d}\lambda \\
    &= \textstyle\int_\Lambda (e^{q(z - z_0)}+ \Phi_{\mathcal{D},\mathcal{B}}(z_0,z,q))e^{qz_0}Q \mathrm{d}\lambda \\
    &= \textstyle\int_\Lambda e^{q(z - z_0)}e^{qz_0}Q \mathrm{d}\lambda
      + \textstyle\int_\Lambda \Phi_{\mathcal{D},\mathcal{B}}(z_0,z,q)e^{qz_0}Q \mathrm{d}\lambda
      \ ,
  \end{align*}
  where
  \begin{align*}
    \textstyle\int_\Lambda e^{q(z - z_0)}e^{qz_0}Q \mathrm{d}\lambda
    = \textstyle\int_\Lambda e^{qz}Q \mathrm{d}\lambda
    = P(z) \ .
  \end{align*}
\end{proof}
To bound the error term isolated in \cref{lem:est-by-diff},
we need to know upper and lower bounds on \(q(\lambda)\) for all \(\lambda \in \Lambda\).
The following lemma functions as \cref{lem:bound-by-eval} for derivatives:
\begin{lemma}
  \label{lem:bound-by-diff}
  Let \(P = P[q, Q]\) and \(v_1,v_2 \in \mathbb{R}\)
  such that \(v_1 \le q(\lambda) \le v_2\) for all \(\lambda \in \Lambda\).
  Then
   \begin{align*}
     |\textstyle\int_\Lambda \Phi_{\mathcal{D},\mathcal{B}}(z_0,z,q)e^{qz_0}Q \mathrm{d}\lambda|
     \le \Phi_{\mathcal{D},\mathcal{B}}^{v_1,v_2}(z_0, z)\textstyle\int_\Lambda e^{qz_0}|Q| \mathrm{d}\lambda
     \ .
  \end{align*}
\end{lemma}
\begin{proof}
  For every \(\lambda \in \Lambda\), we have
  \begin{align*}
    | \Phi_{\mathcal{D},\mathcal{B}}(z_0,z,q(\lambda)) |
    \le \max_{v_1 \le v \le v_2}| \Phi_{\mathcal{D},\mathcal{B}}(z_0,z,v) |
    = \Phi_{\mathcal{D}, \mathcal{B}}^{v_1,v_2}(z_0, z) \ .
  \end{align*}
  Therefore,
  \begin{align*}
    &|\textstyle\int_\Lambda \Phi_{\mathcal{D},\mathcal{B}}(z_0,z,q)e^{qz_0}Q(\lambda) \mathrm{d}\lambda| \\
    &\le \textstyle\int_\Lambda |\Phi_{\mathcal{D},\mathcal{B}}(z_0,z,q)e^{qz_0}Q(\lambda)| \mathrm{d}\lambda \\
    &= \textstyle\int_\Lambda |\Phi_{\mathcal{D},\mathcal{B}}(z_0,z,q)| \, |e^{qz_0}Q(\lambda)| \mathrm{d}\lambda \\
    &\le \textstyle\int_\Lambda \Phi_{\mathcal{D},\mathcal{B}}^{v_1,v_2}(z_0,z)e^{qz_0}|Q| \mathrm{d}\lambda \\
    &= \Phi_{\mathcal{D},\mathcal{B}}^{v_1,v_2}(z_0,z)\textstyle\int_\Lambda e^{qz_0}|Q| \mathrm{d}\lambda
       \ .
  \end{align*}
\end{proof}
Combining \cref{lem:est-by-diff} and \cref{lem:bound-by-diff} into a single statement
yields the derivative analogue of \cref{cor:main1-eval}:
\begin{corollary}
  \label{cor:main1-diff}
  Let \(P = P[q, Q]\),
  \(z_0 \in [0, \infty)\),
  \(\mathcal{D} \subset \mathbb{N}\) finite
  and \(\mathcal{B} = \{\beta_i\}_{i \in \mathcal{D}}\) with \(\beta_i:[0, \infty) \to \mathbb{R}\).
  Furthermore, let \(v_1,v_2 \in \mathbb{R}\)
  such that \(v_1 \le q(\lambda) \le v_2\) for all \(\lambda \in \Lambda\).
  Then
  \begin{align*}
    \textstyle\sum_{i \in \mathcal{D}}\beta_i(z)\frac{\mathrm{d}^iP}{\mathrm{d}z^i}(z_0)
    &= P(z) \pm \Phi_{\mathcal{D},\mathcal{B}}^{v_1,v_2}(z_0, z)\textstyle\int_\Lambda e^{qz_0}|Q| \mathrm{d}\lambda
      \ ,
  \end{align*}
  where the notation \(y = x \pm \delta\) means \(y \in \{x + w\mid |w| \le |\delta|\}\).
\end{corollary}
Finally, we have the derivative version of \cref{thm:limit-eval}:
\begin{theorem}
  \label{thm:limit-diff}
  Let \(P = P[q, Q]\) and \(v_1, v_2 \in \mathbb{R}\) such that \(0 < v_1 \le q(\lambda) \le v_2\) for all \(\lambda \in \Lambda\).
  Let \(\{k_i\}_{i=0}^\infty \subseteq \mathbb{N}\) be an increasing sequence satisfying \(\sum_{i = 1}^\infty k_i^{-1} = \infty\).
  Finally, let \(\mathcal{D}(n) = \{k_i\}_{i=0}^n\), and for any \(z_0, z \in [0, \infty)\) let
  \begin{align*}
    \mathcal{B}(n) = \{\beta_{n, i}\}_{i \in \mathcal{D}(n)} = \arg \min_{\mathcal{B}(n)} \Phi_{\mathcal{D}(n),\mathcal{B}(n)}^{v_1,v_2}(z_0,z)
    \ .
  \end{align*}
  Then
  \begin{align*}
    \lim_{n \to \infty} \textstyle\sum_{i \in \mathcal{D}(n)}\beta_{n, i}(z)\frac{\mathrm{d}^iP}{\mathrm{d}z^i}(z_0) = P(z)
    \ .
  \end{align*}
\end{theorem}
\begin{proof}
  The proof is analogous to that of \cref{thm:limit-eval}.
\end{proof}

\subsection{Final remarks on attenuating functions}

In summary,
attenuating functions form a special class
that is sufficiently restricted so as to enable extrapolation
when \(q\) has known upper and lower bounds.
Sufficient data for this includes evaluations and derivatives,
though other linear functionals are likely to work as well.
The Müntz–Szász theorem provides simple conditions under which
the extrapolation error is guaranteed to vanish in the limit as more data becomes available.
Finally, having an upper bound for \(\int_\Lambda|Q|\mathrm{d}\lambda\)
immediately yields an upper bound for said error, even in the ``pre-asymptotic'' regime.
A noteworthy limitation is that the obtained bounds for the error depend on
\begin{align*}
  \min_{\mathcal{A}(n)} \Psi_{\mathcal{Z}(n), \mathcal{A}(n)}^{u_1,u_2}(z)
  \quad \text{and} \quad
  \min_{\mathcal{B}(n)} \Phi_{\mathcal{D}(n),\mathcal{B}(n)}^{v_1,v_2}(z_0,z)
  \ , 
\end{align*}
which, in the absence of more careful analysis,
can only be determined experimentally on a case-by-case basis.

\section{Using inherent radiance segmentation}
\label{sec:using-irs}

In this section we show how
inherent radiance segmentation reduces color restoration
to extrapolation of attenuating functions,
which we already covered in \cref{sec:att-funcs}.

\begin{definition}
  Given a photograph with pixel intensities
  \begin{align*}
    P_{i,j}(z_i) = \textstyle\int_\Lambda(e^{-cz_i}L_i + (1 - e^{-cz_i})B)S_j\mathrm{d}\lambda
    \ ,
  \end{align*}
  as in \eqref{eq:imf},
  an \emph{inherent radiance segment} is a subset \(\mathcal{I}\) of pixel indices
  satisfying \(L_i = L_\mathcal{I}\) for all \(i \in \mathcal{I}\),
  where \(L_\mathcal{I}\) denotes the shared inherent radiance within \(\mathcal{I}\).
  The set of all inherent radiance segments for the image is denoted by
  \begin{align*}
    P_\triangle
    &= \{
      \mathcal{I} \mid
      L_i = L_\mathcal{I} \ \forall i \in \mathcal{I}
      \} \ .
  \end{align*}  
\end{definition}
For now, we assume that every pixel \(i\)
belongs to some known segment \(\mathcal{I} \in P_\triangle\).
Obtaining \emph{some} segment is trivial: simply pick \(\mathcal{I} = \{i\}\),
though this will yield but a single point for attenuating function extrapolation,
resulting in poor error bounds.
On the other hand, obtaining non-trivial segments is exactly that\textemdash
non-trivial\textemdash
and in the fully general setting it is ill-posed.
Instead of showing this by fully describing the solution space
like we did in \cref{sec:ill-posedness} for color restoration,
we resort to merely mentioning metamerism.
Indeed, it is well-known that even in clear air,
two pixels of identical color need not correspond to the same (inherent) radiance,
making it impossible to determine whether these two pixels
should belong to the same segment.
This naturally raises the question of whether
achieving non-trivial inherent radiance segmentation is a hopeless endeavor;
however, as we shall see in \cref{sec:obtaining-irs},
it is theoretically possible in a sufficiently constrained, albeit highly idealized, setting.

In any case,
the utility of such segmentation
comes from the fact that the pixel intensities within each segment \(\mathcal{I}\)
lose their explicit dependence on \(i\)
and instead become dependent only on \(z_i\),
bringing us closer to the realm of attenuating functions.
To capture this in notation, we rely on the following definition,
which also includes some auxiliary quantities:
\begin{definition}
  \label{def:Pbar}
  For any \(\mathcal{I} \in P_\triangle\) and any channel \(j\), define
  \begin{align*}
    P_{\mathcal{I}, j}(z)
    &= \textstyle\int_\Lambda(e^{-cz}L_{\mathcal{I}} + (1 - e^{-cz})B)S_j\mathrm{d}\lambda \ , \\  
    P_j(\infty)
    &= \textstyle\int_\Lambda BS_j\mathrm{d}\lambda
      = \lim_{z \to \infty} P_{i,j}(z) \quad \forall i \ , \\
    \bar{P}_{\mathcal{I},j}(z)
    &= P_{\mathcal{I},j}(z) - P_j(\infty)
      = \textstyle\int_\Lambda e^{-cz}(L_{\mathcal{I}} - B)S_j\mathrm{d}\lambda
      \ .
  \end{align*}
\end{definition}
For any \(\mathcal{I} \in P_\triangle\),
the function \(P_{\mathcal{I}, j}: [0, \infty) \to [0, \infty)\)
describes the pixel intensity associated with the inherent radiance \(L_\mathcal{I}\)
at any target distance \(z \in [0, \infty)\).
The quantity \(P_j(\infty)\) is the pixel intensity ``at infinity'', or ``at the horizon''.
Whenever \(P_j(\infty)\) is known, we can easily convert back and forth between
the original pixel intensities \(P_{\mathcal{I}, j}(z_i)\)
and the ``reduced'' pixel intensities \(\bar{P}_{\mathcal{I},j}(z_i)\),
which is helpful precisely because \(\bar{P}_{\mathcal{I},j}\) is an attenuating function.
In other words, the task of estimating \(P_{\mathcal{I},j}(z)\)
from the \(z_i\) and the \(P_{\mathcal{I},j}(z_i)\) for \(i \in \mathcal{I}\),
becomes equivalent to that of estimating
\(\bar{P}_{\mathcal{I},j}(z)\) from the \(z_i\) and the \(\bar{P}_{\mathcal{I},j}(z_i)\).
This is a canonical case of attenuating function extrapolation from evaluations
described in \cref{subsec:evaluations}.
The only additional data required are upper and lower bounds
on the beam attenuation coefficient \(c\) across the visible spectrum \(\Lambda\).
The following is a specialization of \cref{cor:main1-eval} to the color restoration setting:
\begin{theorem}
  \label{thm:main1-eval}
  Let  \(u_1,u_2 \in \mathbb{R}\) with
  \(u_1 \le e^{-c(\lambda)} \le u_2\) for all \(\lambda \in \Lambda\).
  Let \(\mathcal{I} \in P_\triangle\),
  \(\mathcal{Z}(\mathcal{I}) = \{z_i\}_{i \in \mathcal{I}} \subset [0, \infty)\),
  and let \(\mathcal{A}(\mathcal{I}) = \{\alpha_i\}_{i \in \mathcal{I}}\) with \(\alpha_i: [0, \infty) \to \mathbb{R}\).
  Then
  \begin{align*}
    \bar{P}_{\mathcal{I},j}(z) = \textstyle\sum_{i \in \mathcal{I}}\alpha_i(z)\bar{P}_{\mathcal{I},j}(z_i)
      \ \pm \
      \Psi_{\mathcal{Z}(\mathcal{I}), \mathcal{A}(\mathcal{I})}^{u_1,u_2}(z)\textstyle\int_\Lambda |Q_{\mathcal{I},j}| \mathrm{d}\lambda
      ,
  \end{align*}
  where \(Q_{\mathcal{I},j} = (L_\mathcal{I} - B)S_j\).
\end{theorem}
\begin{proof}
  This is \cref{cor:main1-eval} with \(\bar{P}_{\mathcal{I},j} = P[-c, Q_{\mathcal{I},j}]\).
\end{proof}
The practical meaning of \cref{thm:main1-eval} is that choosing
\(\mathcal{A}(\mathcal{I}) = \{\alpha_i\}_{i \in \mathcal{I}}\) so as to minimize \(\Psi_{\mathcal{Z}(\mathcal{I}), \mathcal{A}(\mathcal{I})}^{u_1,u_2}(z)\)
for some target distance \(z\),
yields the estimate
\begin{align*}
  P_{\mathcal{I},j}(z)
  &= \bar{P}_{\mathcal{I},j}(z) + P_j(\infty) \\
  &\approx \textstyle\sum_{i \in \mathcal{I}}\alpha_i(z)\bar{P}_{\mathcal{I},j}(z_i) + P_j(\infty)
    \ .
\end{align*}
Furthermore, if we somehow know an upper bound on
\begin{align*}
  \textstyle\int_\Lambda |Q_{\mathcal{I},j}| \mathrm{d}\lambda
  &= \textstyle\int_\Lambda |(L_\mathcal{I} - B)S_j| \mathrm{d}\lambda
    \ ,
\end{align*}
then we immediately get an upper bound for the estimation error as well.
One way of obtaining such a bound from
individual bounds on \(L_\mathcal{I}\), \(B\) and \(S_j\),
is to use Hölder's inequality together with the triangle inequality:
\begin{align*}
  \textstyle\int_\Lambda |(L_\mathcal{I} - B)S_j| \mathrm{d}\lambda
  &\le \lVert L_\mathcal{I} - B \rVert_p \lVert S_j \rVert_q \\
  &\le (\lVert L_\mathcal{I} \rVert_p + \lVert B \rVert_p) \lVert S_j \rVert_q
  \ ,
\end{align*}
for any \(p,q \in [1, \infty]\) with \(1/p + 1/q = 1\).
Regardless of what this bound is,
we know from \cref{thm:limit-eval} that
the estimation error is guaranteed to vanish
in the limit \(|\mathcal{Z}(\mathcal{I})| \to \infty\),
provided \(\sum_{z_i \in \mathcal{Z}(\mathcal{I})}z_i^{-1} \to \infty\).
In simpler terms: once we have identified a region in the photograph
that corresponds to an inherent radiance segment,
then the estimation error for \(P_{\mathcal{I},j}(z)\)
can be made arbitrarily small
by increasing the camera resolution,
provided the region keeps resolving new distances.
The constraint \(\sum_{z_i \in \mathcal{Z}(\mathcal{I})}z_i^{-1} \to \infty\)
will generally be satisfied for reasonable regions,
e.g., if there is some maximal distance \(z_{\text{max}}\)
within such a region, then
\begin{align*}
  \textstyle\sum_{z_i \in \mathcal{Z}(\mathcal{I})}z_i^{-1} \ge \sum_{z_i \in \mathcal{Z}(\mathcal{I})}z_{\text{max}}^{-1} = z_{\text{max}}^{-1}\sum_{z_i \in \mathcal{Z}(\mathcal{I})}1 \to \infty
  \ .
\end{align*}
The conclusion to draw from this section is that
the fatal ill-posedness of color restoration in the fully general setting
can be meaningfully resolved
under the additional constraint of known inherent radiance segmentation.
The bad news is that obtaining such segmentation
is necessarily as ill-posed as the original problem.
In \cref{sec:obtaining-irs}, we identify sufficient conditions
for guaranteeing a unique algorithmic solution.

\section{Obtaining inherent radiance segmentation}
\label{sec:obtaining-irs}

In \cref{sec:using-irs},
we saw how color restoration can be reduced to inherent radiance segmentation.
Here, we give sufficient conditions under which the latter
can be determined algorithmically.
Instead of stating these conditions for a single image
viewed as a finite list of pixel values,
it will be more convenient to make asymptotic statements
about a sequence of images with increasingly high resolution
for the same underlying scene.
To model this formally,
we consider a hypothetical continuum of pixel values
associated with each point in the unit square \([0, 1]^2\),
which is then sampled by an infinite coordinate sequence
\(\{(x_i, y_i)\}_{i=1}^\infty\) that is dense in \([0, 1]^2\).
Any finite-resolution image may then be regarded
as a truncation of this infinite process by restricting \(1 \le i \le n\)
for some positive integer \(n\),
allowing us to reason about the regime \(n \to \infty\).
Furthermore, we immediately consider the ``reduced'' pixel intensities
obtained by subtracting \(P_j(\infty)\) from the entire image, as in \cref{def:Pbar},
though this is mostly for cleaner notation.
Finally, we omit the channel index \(j\),
since cross-channel information is never used,
and we write \(q = -c\) and \(Q[x, y] = (L[x, y] - B)S_j\),
where \(L[x, y]: \Lambda \to \mathbb{R}\) is the inherent radiance associated with the point \((x, y) \in [0, 1]^2\).

\begin{definition}
  Define the \emph{profile} \(p: [0, 1]^2 \to \mathbb{R}\) by 
  \begin{align*}
    p: (x, y) \mapsto \textstyle\int_\Lambda e^{qz(x, y)}Q[x, y]\mathrm{d}\lambda \ ,
  \end{align*}
  where \(q, Q[x, y] : \Lambda \to \mathbb{R}\) and \(z: [0, 1]^2 \to [0, \infty)\).
  Also, let
  \begin{align*}
    g(x, y) &= (z(x, y), p(x, y))
              \ , \\
    \mathcal{H}(x, y) &= \{
              (x', y') \mid Q[x', y'] = Q[x, y]
              \} \ .
  \end{align*}
  Finally, let
  \(\{(x_i,y_i)\}_{i=1}^\infty\) be a sequence dense in \([0, 1]^2\).
\end{definition}

Intuitively, \(p(x, y)\) is the reduced pixel intensity at the point \((x, y)\) in the image,
while \(z(x, y)\) is the associated distance to the scene.
The pair \(g(x, y) = (z(x, y), p(x, y))\) is merely a convenient datastructure,
while \(\mathcal{H}(x, y) \subset [0, 1]^2\) is the region in the image where
the inherent radiance is the same as at the point \((x, y)\).
This region may be conceptualized as the continuous analogue of an inherent radiance segment,
and it is maximal in the sense of containing all points \((x', y')\) with \(Q[x', y'] = Q[x, y]\).
We proceed by defining certain constraints on \(g\) and \(\mathcal{H}\)
that are involved in guaranteeing well-posed inherent radiance segmentation.
For this, we rely on a few basic concepts from topology:
\begin{definition}
  We will say that \(p\) is \emph{connected} if
  \(\mathcal{H}(x, y)\) is always connected.
  We will say that \(p\) is \emph{discrete} if
  \(\mathcal{H}(x, y)\) is always semi-open,
  i.e., its interior is dense in its closure.
  Finally, for any \(\varepsilon > 0\), we say that \(p\) is \(\varepsilon\)-\emph{separable} if
  \begin{align*}
    Q[x, y] \ne Q[x', y']
    \implies
    \lVert g(x, y) - g(x', y') \rVert \ge \varepsilon
    \ .
  \end{align*}
\end{definition}

Intuitively, \(p\) being connected
means that each continuous inherent radiance segment \(\mathcal{H}(x, y)\)
forms a single ``island'', instead of an ``archipelago''.
This constraint is not essential, but we keep it for simplicity.
Having a discrete \(p\) forces these ``islands'' to be two-dimensional,
avoiding many pathologies,
and it also forbids the spatial gradient of \(Q\) for every \(\lambda\)
to be non-zero where it is defined.
This admittedly strong assumption is central to our argument,
and future research should focus on relaxing it.
Finally, \(\varepsilon\)-separability of \(p\) is designed to take care of
any potential issues associated with metamerism.
Some form of this is likely necessary,
unless working with a perfect hyperspectral camera.
The main result of this section is given in \cref{thm:main2},
but to get there we prove a few basic facts about \(g(\mathcal{H}(x, y))\).
\begin{lemma}
  \label{lem:g(H)-connected}
  If \(z\) is continuous and \(p\) is connected,
  then \(g(\mathcal{H}(x, y))\) is connected.
\end{lemma}
\begin{proof}  
  Since \(z\) is continuous on the entire \([0, 1]^2\),
  then it is also continuous on \(\mathcal{H}(x, y) \subseteq [0, 1]^2\).
  Since \(p\) is also continuous on \(\mathcal{H}(x, y)\),
  then so is \(g: (x, y) \mapsto (z(x, y), p(x, y))\).
  But \(\mathcal{H}(x, y)\) is connected as a consequence of \(p\) being connected,
  and the image of a continuous map on a connected space is always connected,
  hence \(g(\mathcal{H}(x, y))\) is connected.
\end{proof}

\begin{lemma}
  \label{lem:g(H)-bounded}
  If \(z\) is continuous,
  then \(g(\mathcal{H}(x, y))\) is bounded.
\end{lemma}
\begin{proof}
  Since \([0, 1]^2\) is compact, then the closure \(\overline{\mathcal{H}(x, y)}\) of \(\mathcal{H}(x, y) \subseteq [0, 1]^2\)
  is also compact.
  Since \(z\) is continuous on \([0, 1]^2\),
  then it is also continuous on \(\overline{\mathcal{H}(x, y)}\),
  hence \(z(\overline{\mathcal{H}(x, y)})\) is compact and therefore bounded.
  Furthermore, for any \((x', y') \in \mathcal{H}(x, y)\),
  the Cauchy-Schwarz inequality yields
  \begin{align*}
    p(x', y')^2
    &\le \textstyle \int_\Lambda e^{2qz(x', y')} \mathrm{d}\lambda \textstyle \int_\Lambda Q[x, y]^2 \mathrm{d}\lambda
    \ .
  \end{align*}
  Since \(z(\mathcal{H}(x, y))\) is bounded,
  then the right-hand side of the above inequality is also bounded,
  which implies that \(p(\mathcal{H}(x, y))\) is bounded.
  But since
  \begin{align*}
    g(\mathcal{H}(x, y)) \subseteq z(\mathcal{H}(x, y)) \times p(\mathcal{H}(x, y))
    \ ,
  \end{align*}
  then \(g(\mathcal{H}(x, y))\) is bounded as well.
\end{proof}
Now, we prove a few facts about \(\{g(x_i, y_i)\}_{i=1}^\infty\).
\begin{lemma}
  \label{lem:g-dense}
  Suppose \(z\) is continuous and \(p\) is discrete.
  Then \(\{g(x_i, y_i)\}_{i=1}^\infty\)
  has a subsequence dense in \(g(\mathcal{H}(x, y))\).
\end{lemma}
\begin{proof}
  Since \(\{(x_i, y_i)\}_{i=1}^\infty\) is dense in \([0, 1]^2\)
  and \(\mathcal{H}(x, y)\) is semi-open,
  then there is some subsequence \(\{(x_i, y_i)\}_{i \in \mathcal{J}}\),
  with \(\mathcal{J} \subset \mathbb{N}\),
  that is dense in \(\mathcal{H}(x, y)\).
  Therefore, if \(g\) is continuous on \(\mathcal{H}(x, y)\),
  then \(\{g(x_i, y_i)\}_{i \in \mathcal{J}}\) is dense in \(g(\mathcal{H}(x, y))\).
  But \(z\) and \(p\) are both continuous on \(\mathcal{H}(x, y)\),
  hence so is \(g\).
\end{proof}
\begin{lemma}
  \label{lem:cover}
  Suppose \(z\) is continuous and \(p\) is discrete.
  Then for any \(\varepsilon > 0\), there exists a finite \(\mathcal{I} \subset \mathbb{N}\) satisfying
  \begin{align*}
    \{(x_i,y_i)\}_{i \in \mathcal{I}}
    &\subset \mathcal{H}(x, y) \ , \\
    g(\mathcal{H}(x, y))
    &\subset \cup_{i \in \mathcal{I}} \mathcal{B}_\varepsilon(g(x_i, y_i))
    \ ,
  \end{align*}
  where \(\mathcal{B}_\varepsilon(g(x_i, y_i)) \subset \mathbb{R}^2\) denotes the open ball of radius \(\varepsilon\) centered on \(g(x_i, y_i)\).
\end{lemma}
\begin{proof}
  Since \(z\) is continuous, then \(g(\mathcal{H}(x, y))\) is bounded due to \cref{lem:g(H)-bounded}.
  Since \(g(\mathcal{H}(x, y)) \subset \mathbb{R}^2\), then it is totally bounded
  and can thus be covered by finitely many open balls of radius \(\varepsilon/2\).
  Since \(p\) is discrete and \(z\) is continuous,
  then \(\{g(x_i,y_i)\}_{i \in \mathcal{J}}\) is dense in \(g(\mathcal{H}(x, y))\)
  for some \(\mathcal{J} \subset \mathbb{N}\) due to \cref{lem:g-dense}.
  Consequently, each of the aforementioned \(\varepsilon/2\)-balls
  contains an element of \(\{g(x_i,y_i)\}_{i \in \mathcal{J}}\),
  and the same also holds for some \(\{g(x_i,y_i)\}_{i \in \mathcal{I}}\),
  where \(\mathcal{I} \subset \mathcal{J}\) is finite.
  Taking an open ball of radius \(\varepsilon\)
  around each element of \(\{g(x_i,y_i)\}_{i \in \mathcal{I}}\),
  we see that these \(\varepsilon\)-balls,
  in their union,
  cover all of the previously mentioned \(\varepsilon/2\)-balls.
  Consequently, these \(\varepsilon\)-balls cover \(g(\mathcal{H}(x, y))\) as well.
\end{proof}
\begin{lemma}
  \label{lem:intersection}
  Suppose \(z\) is continuous and \(p\) is connected.
  Let \(\varepsilon > 0\), and suppose some \(\mathcal{I} \subset \mathbb{N}\) satisfies
  \begin{align*}
    \{(x_i,y_i)\}_{i \in \mathcal{I}}
    &\subset \mathcal{H}(x, y) \ , \\  
    g(\mathcal{H}(x, y))
    &\subset \cup_{i \in \mathcal{I}} \mathcal{B}_\varepsilon(g(x_i, y_i))
    \ .
  \end{align*}
  Then any non-empty proper subset \(\varnothing \ne \mathcal{J} \subsetneq \mathcal{I}\) satisfies
  \begin{align*}
    \cup_{i \in \mathcal{J}} \mathcal{B}_\varepsilon(g(x_i, y_i))
    \ \cap \
    \cup_{i \in \mathcal{I} \setminus \mathcal{J}} \mathcal{B}_\varepsilon(g(x_i, y_i))
    \ \ne \ \varnothing
    \ .
  \end{align*}
\end{lemma}
\begin{proof}
  Since \(z\) is continuous and \(p\) is connected,
  then \(g(\mathcal{H}(x, y))\) is also connected due to \cref{lem:g(H)-connected}.
  Since \(\{(x_i,y_i)\}_{i \in \mathcal{I}} \subset \mathcal{H}(x, y)\),
  then
  \begin{align*}
    g_1 &\doteq g(\mathcal{H}(x, y)) \ \cap \ \cup_{i \in \mathcal{J}} \mathcal{B}_\varepsilon(g(x_i, y_i)) \ \ne \ \varnothing \ , \\
    g_2 &\doteq g(\mathcal{H}(x, y)) \ \cap \ \cup_{i \in \mathcal{I} \setminus \mathcal{J}} \mathcal{B}_\varepsilon(g(x_i, y_i)) \ \ne \ \varnothing \ .
  \end{align*}
  Since \(g_1\) and \(g_2\) are open in \(g(\mathcal{H}(x, y))\)
  with respect to the subspace topology induced by the usual topology on \(\mathbb{R}^2\), and
  \begin{align*}
    g_1 \cup g_2
    &= g(\mathcal{H}(x, y)) \ \cap \ \cup_{i \in \mathcal{I}} \mathcal{B}_\varepsilon(g(x_i, y_i))
    \ = \ g(\mathcal{H}(x, y))
    \ , \\
    g_1 \cap g_2
    &\subset \cup_{i \in \mathcal{J}} \mathcal{B}_\varepsilon(g(x_i, y_i)) \ \cap \ \cup_{i \in \mathcal{I} \setminus \mathcal{J}} \mathcal{B}_\varepsilon(g(x_i, y_i)) \ ,
  \end{align*}
  then \(g_1 \cap g_2 \ne \varnothing\) due to \(g(\mathcal{H}(x, y))\) being connected.
  But then
  \begin{align*}
    \cup_{i \in \mathcal{J}} \mathcal{B}_\varepsilon(g(x_i, y_i)) \ \cap \ \cup_{i \in \mathcal{I} \setminus \mathcal{J}} \mathcal{B}_\varepsilon(g(x_i, y_i)) \ \ne \ \varnothing \ .
  \end{align*}
\end{proof}
We are now ready to combine
\cref{lem:cover} and \cref{lem:intersection}
to show that inherent radiance segmentation
is possible to obtain algorithmically in idealized conditions:
\begin{theorem}
  \label{thm:main2}
  For any \((x, y) \in [0, 1]^2\), let
  \begin{align*}
    p: (x, y) \mapsto \textstyle\int_\Lambda e^{qz(x, y)}Q[x, y]\mathrm{d}\lambda \ ,
  \end{align*}
  where \(q, Q[x, y] : \Lambda \to \mathbb{R}\) and \(z: [0, 1]^2 \to [0, \infty)\).
  Suppose \(z\) is continuous,
  while \(p\) is discrete, connected and \(\varepsilon\)-separable.
  Assume \(\{(x_i, y_i)\}_{i=1}^\infty\) is a dense sequence in \([0, 1]^2\).
  Then there exists an algorithm which inputs
  \(z(x_i, y_i)\) and \(p(x_i, y_i)\)
  for \(i = 1,\dots,n\),
  and correctly outputs,
  for \(k = 1,\dots, n\),
  the sets
  \begin{align*}
    \mathcal{I}_n(k) = \{ 1 \le i \le n \mid Q[x_i, y_i] = Q[x_k, y_k] \}
    \ ,
  \end{align*}
  provided \(n\) is large enough and \(\varepsilon\) is known (or lower-bounded).
\end{theorem}
\begin{proof}  
  Provided \(n\) is large enough,
  since \(z\) is continuous and \(p\) is discrete,
  then \cref{lem:cover} implies that there is some
  \(\mathcal{I} \subseteq \{1, \dots, n\}\) satisfying
  \begin{align*}
    \{(x_i,y_i)\}_{i \in \mathcal{I}}
    &\subset \mathcal{H}(x_k, y_k) \ , \\
    g(\mathcal{H}(x_k, y_k))
    &\subset \cup_{i \in \mathcal{I}} \mathcal{B}_{\varepsilon/2}(g(x_i, y_i))
    \ .
  \end{align*}
  But since
  \begin{align*}
    \mathcal{I}_n(k) = \{1 \le i \le n \mid (x_i, y_i) \in \mathcal{H}(x_k, y_k)\} \ ,
  \end{align*}
  then \(\mathcal{I} \subseteq \mathcal{I}_n(k)\) and
  \begin{align*}
    g(\mathcal{H}(x_k, y_k)) \subset \cup_{i \in \mathcal{I}_n(k)} \mathcal{B}_{\varepsilon/2}(g(x_i, y_i))
    \ .
  \end{align*}
  Since \(z\) is continuous and \(p\) is connected,
  then \cref{lem:intersection} implies that any \(\mathcal{J}\) with
  \(\varnothing \ne \mathcal{J} \subsetneq \mathcal{I}_n(k)\)
  satisfies
  \begin{align*}
    \cup_{i \in \mathcal{J}} \mathcal{B}_{\varepsilon/2}(g(x_i, y_i))
    \ \cap \
    \cup_{i \in \mathcal{I}_n(k) \setminus \mathcal{J}} \mathcal{B}_{\varepsilon/2}(g(x_i, y_i))
    \ \ne \ \varnothing
    \ .
  \end{align*}  
  But this can happen only if
  \begin{align*}
    \min_{a \in \mathcal{J}, \, b \in \mathcal{I}_n(k) \setminus \mathcal{J}} \lVert g(x_a, y_a) - g(x_b, y_b) \rVert < \varepsilon
    \ .
  \end{align*}
  On the other hand, since \(p\) is \(\varepsilon\)-separable, then
  \begin{align*}
    \min_{a \in \mathcal{J}, \, b \not \in \mathcal{I}_n(k)} \lVert g(x_a, y_a) - g(x_b, y_b) \rVert \ge \varepsilon
    \ ,
  \end{align*}
  where \(b \not \in \mathcal{I}_n(k)\) is shorthand for \(b \in \{1, \dots, n\} \setminus \mathcal{I}_n(k)\).
  Now, let \(\mathcal{J}_1 = \{k\} \subseteq \mathcal{I}_n(k)\), and for any \(1 \le t < n\) let
  \begin{align*}
    \mathcal{J}_{t+1} = \mathcal{J}_t \cup \{\arg \min_{b \not \in \mathcal{J}_t} \min_{a \in \mathcal{J}_t}\lVert g(x_a, y_a) - g(x_b, y_b) \rVert\}
    \ .
  \end{align*}
  Then, by induction, we have
  \begin{align*}
    \mathcal{J}_1 \subset \cdots \subset \mathcal{J}_T = \mathcal{I}_n(k)
    \ ,
  \end{align*}
  where \(T = |\mathcal{I}_n(k)|\) is the smallest integer satisfying
  \begin{align*}
    \min_{b \not \in \mathcal{J}_T}\min_{a \in \mathcal{J}_T} \lVert g(x_a, y_a) - g(x_b, y_b) \rVert \ge \varepsilon
    \ .
  \end{align*}
  This fully defines an algorithm for obtaining \(\mathcal{I}_n(k)\),
  provided \(\varepsilon\) is known,
  though any positive lower bound works as well.
\end{proof}

A visualization of \cref{thm:main2} on a simple example
is given in \cref{fig:segmentation}.
The associated plot for \(g(x, y)\)
highlights what makes inherent radiance segmentation tractable in this setting:
the fact that \(g(\mathcal{H}(x, y))\) is a continuous curve for any point \((x, y)\).
Perhaps the greatest deviation from a real setting here is the disallowance of shadow gradients
as a consequence of focusing on the rather strict notion of ``regions'' given by
\begin{align*}
  \mathcal{H}(x, y) = \{(x', y') \mid Q[x', y'] = Q[x, y]\} \ .
\end{align*}
This should not be too hard to generalize, however,
for example by considering the more ``relaxed'' version:
\begin{align*}
  \mathcal{H}(x, y) = \{(x', y') \mid \exists \gamma \in \mathbb{R} \setminus \{0\} \text{ s.t. } \gamma Q[x', y'] = Q[x, y]\} \ .
\end{align*}
This would result in the sets \(g(\mathcal{H}(x, y))\)
being two-dimensional regions rather than one-dimensional curves,
which should not cause too much trouble in identifying them
provided the spatial resolution of the camera is sufficiently high.
It would, however, introduce new challenges in adapting
the results from \cref{sec:att-funcs}
to this more relaxed setting,
which we leave for future research.

\begin{figure}[htb]
  \centering
  \includegraphics[width=1.0\linewidth]{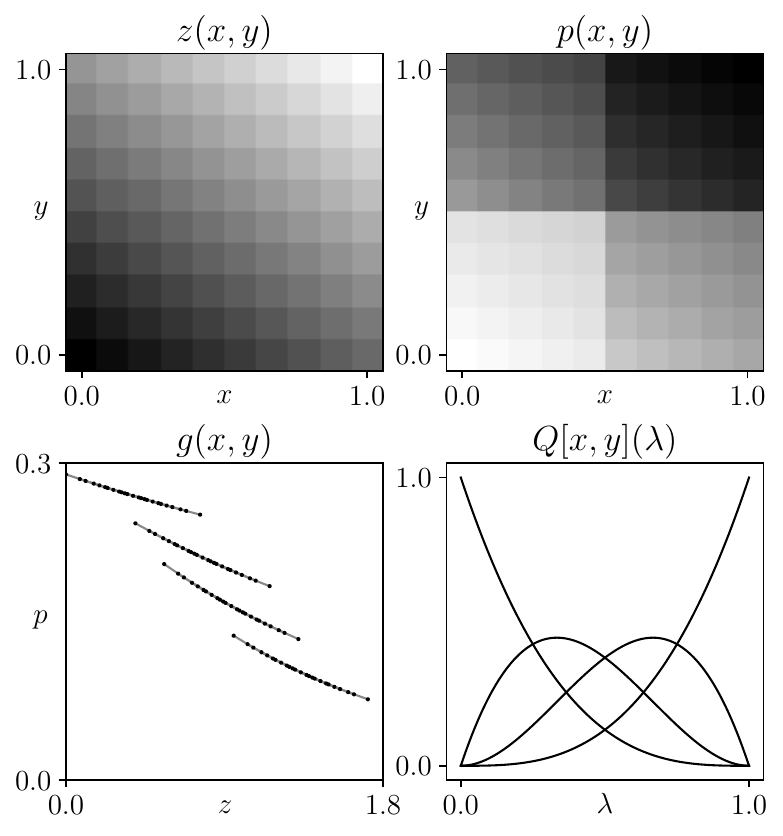}
  \caption{
    Inherent radiance segmentation in an idealized setting.
    Top left: depth map \(z(x, y) = y + x/\sqrt{2}\) visualized as a heat map.
    Top right: pixel intensities \(p(x, y) = \int_\Lambda e^{qz(x,y)}Q[x, y]\mathrm{d}\lambda\),
    where \(\Lambda = [0, 1], q(\lambda) = -\lambda\), and \(Q[x, y](\lambda) = \binom{3}{k}\lambda^k(1-\lambda)^{3-k}\) for \(0 \le k \le 3\)
    are the 4 Bernstein polynomials of degree 3.
    Bottom left: pairs \(g(x, y) = (z(x, y), p(x, y))\),
    with the black points coming from the pixels
    and the gray curves coming from the regions \(\mathcal{H}(x, y)\).
    Bottom right: The Bernstein polynomials used for the functions \(Q[x, y] = (L[x, y] - B)S_j\).
    As the image resolution increases,
    the black points in the bottom left subfigure become
    increasingly concentrated within their associated gray curves,
    making segmentation easier.
  }
  \label{fig:segmentation}
\end{figure}

\section{Discussion and conclusion}
\label{sec:conclusion}

The difficulty of validating underwater color restoration methods
is the greatest obstacle to using color as a reliable signal in aquatic sciences.
The vast diversity of possible visibility conditions
poses practical challenges for empirical validation,
while theoretical validation lacks mathematical tools
for understanding when the underlying problem is well-posed.
As an early contribution to the latter,
we investigated the utility of inherent radiance segmentation,
which lends itself well to mathematical analysis due to being simple to state.
We found that having oracle access to such segmentation
is often sufficient to narrow the correct solution to a finite uncertainty interval
whose length converges to zero as the spatial resolution of the camera increases.
Furthermore, we identified the first conditions under which
such segmentation can be determined algorithmically,
though these are too idealized to be immediately applicable to real-world data,
requiring inherent radiance to be piecewise constant in image coordinates.
Still, understanding this simplified setting has illuminated
possible paths forward for future generalizations,
including that of allowing continuously varying shadows.

In conclusion,
the gap in underwater color restoration
between what we know about the constraints of real-world data
and the constraints that guarantee solvability remains significant,
though the latter has now been demonstrated to be susceptible to
mathematical investigation.


%

\appendices


\section*{Acknowledgments}
We thank Dr. Herdís Steinsdóttir for assistance with illustrations, and all COLOR Lab members for support.
\paragraph*{Funding}
This work was supported by grants to D.A. from the Schmidt Marine Technology Partners (G-22-63208, G-24-66610), Israel Science Foundation (1055/22, 2788/22), Office of Naval Research Global (N629092312104), and European Union’s Horizon 2020 research and innovation program GA (101094924, ANERIS). G.S. was supported by a postdoctoral scholarship from the Maurice Hatter Foundation. 
\paragraph*{Author contributions}
Authors contributed equally to this work.
\paragraph*{Competing interests}
There are no competing interests to declare.

\ifCLASSOPTIONcaptionsoff
  \newpage
\fi



%
%
%

\bibliographystyle{IEEEtran}
\bibliography{biblio}

%








\end{document}